\documentclass[a4paper,11pt, notitlepage]{article}
\usepackage[T1]{fontenc}
\usepackage[utf8]{inputenc}
\usepackage{listings}
\usepackage{geometry}
\usepackage{graphicx}
\usepackage{amsfonts}
\usepackage{tikz}
\usepackage{hyperref}
\usepackage{amsmath, amsthm}
\usepackage[linesnumbered, algoruled]{algorithm2e}

\title{Fault-tolerant parallel scheduling of 
arbitrary
length jobs on a shared channel}

\author{Marek Klonowski\footnote{Faculty of Fundamental Problems of Technology,
		Wrocław University of Technology, Wybrzeże Wyspiańskiego 27, 50-370 Wrocław, Poland.
		Email: \texttt{Marek.Klonowski@pwr.edu.pl}.} \and 
Dariusz R. Kowalski\footnote{Department of Computer Science,
	University of Liverpool, Ashton Building, Ashton Street,
	Liverpool L69 3BX, UK. 
	Email: \texttt{\{d.kowalski,j.mirek,pwong\}@liverpool.ac.uk}.
	\pw{The work is partly supported by the initiative Networks and Sciences Techologies (NeST), School of EEE\&CS, University of Liverpool.}} \and
Jarosław Mirek\footnotemark[2] \and
Prudence W.H. Wong\footnotemark[2]}

\date{}

\newtheorem{lemma}{Lemma}
\newtheorem{theorem}{Theorem}
\newtheorem{fact}{Fact}
\newtheorem{corollary}{Corollary}

\newtheorem{claim}{Claim}
\newtheorem*{claim*}{Claim}

\newcommand{\cO}{{\mathcal O}}

\newcommand{\cE}{{\mathcal E}}

\newcommand{\fba}{\textit{$f$-bounded}}

\newcommand{\ua}{\textit{unbounded}}

\newcommand{\tpbb}{\textsc{TaPeBB}}
\newcommand{\algx}{\textsc{ScaTri}}
\newcommand{\algy}{\textsc{DefTri}}
\newcommand{\algz}{\textsc{RanScaTri}}

\newcommand{\mat}{\textsc{Mix-And-Test}}
\newcommand{\conwor}{\textsc{Confirm-Work}}

\newcommand{\DA}{\textit{Do-All}}

\newcommand{\remove}[1]{}

\newcommand{\dk}[1]{#1}
\newcommand{\pw}[1]{#1}
\newcommand{\jm}[1]{#1}

\begin{document}

\maketitle

\begin{abstract}
We study the problem of scheduling jobs on fault-prone machines communicating via a shared channel without collision detection, also known as a multiple-access channel. We have $n$ arbitrary length jobs to be scheduled on $m$ identical machines, 
$f$ of which are prone to crashes by an adversary. A machine can inform other machines when a job is completed via the channel. 
Performance is measured by the total number of available machine steps during the whole execution. 
Our goal is to study the impact of preemption (i.e., interrupting the execution of a job and resuming later in the same or different machine) and failures on the work performance
of job processing.
The novelty is the ability to identify the features that 
determine the complexity (difficulty) of the problem. 
We show that the problem becomes difficult when preemption is not allowed,
by showing corresponding lower and upper bounds, the latter with algorithms
reaching them. 
We also prove that randomization helps even more, but only against a non-adaptive
adversary; in the presence of more severe adaptive adversary, 
randomization does not help in any setting.  
Our work has extended from previous work that focused on settings including: scheduling on multiple-access channel without machine failures, complete information about failures, or incomplete information 
about failures (like in this work) but with unit length jobs and, hence, without considering preemption. 
\end{abstract}

\setcounter{page}{0}
\thispagestyle{empty}

\pagebreak

\section{Introduction}
\label{sec01}

We examine the problem of performing $ n $ jobs by $ m $ machines reliably on a multiple-accessed shared channel (MAC). This problem in such model, originated by Chlebus, Kowalski and Lingas \cite{CKL}, has already 
been studied for unit length jobs,
whereas this paper extends it by considering jobs with arbitrary lengths,
studying the impact of features including preemption, randomization, and severity of failures on the work performance.

The notion of preemption may be understood as the possibility of not performing a particular job in one attempt. This means that a single job may be interrupted partway through performing it 
and then be resumed later by the same machine or even by another machine. Intuitively, the model without preemption is a more general configuration, yet both have subtleties that distinguish them noticeably.
In general, algorithms exploiting randomization may improve the performance and we aim to investigate when randomization actually helps.

The communication takes place on a shared channel, also known as multiple access channel, without collision detection. 
\pw{A message is transmitted successfully only when one machine transmits, and
when more than one message is transmitted then it is no \dk{different from background noise (i.e., no transmission on the channel).}
Therefore, we say that a job completion is confirmed only when the corresponding machine transmits such message successfully (before any failure it may suffer).}
We consider the impact of the severity of adversaries (measuring the severity of machine failures) on the performance.
In particular, we distinguish between adaptive and non-adaptive adversaries.
An \textit{adaptive $f$-bounded} adversary 
may decide to crash up to $f$ machines 
at any time of the execution.
A \emph{non-adaptive} adversary has to additionally \jm{decide} 
in the beginning of the execution which, up to $f$, machines and when to crash.
%
%
We use work as the effectiveness measure, i.e., the total number of machine steps available for computations. This measure is related (after division by $m$)
to the average time performance of available machines. 
\jm{Work may also correspond to energy consumption, as an operational machine generates (consumes) a unit of work (energy) in every time step.}

\subsection{Previous and related results}

Our work can be seen as an extension of the \DA\ problem defined in the seminal work by Dwork, Halpern and Waarts~\cite{DHW}.
\jm{This line of research was} followed up by several other papers~\cite{CDS,CGKS,CK,DMY,GMY} which considered the message-passing model with point-to-point communication. 
In all these papers the authors assumed that the performance of a single job contributes a unit to the work complexity. Paper~\cite{DMY} introduced a model, wherein the measure of work for the \DA\ solutions  
is the available machine steps (i.e., including idle rounds). They developed an algorithm which has work $\cO(n+(f+1)m)$  and message complexity $\cO((f+1)m)$. 
We also adopt this work measurement in our paper.  

In~\cite{GMY} the authors improved the message complexity to $\cO(f m^\varepsilon + \min\{f+1, \log m\}m)$, for any $\varepsilon >0 $ keeping  the same work complexity.
 \DA\ problem with failures controlled by a weakly-adaptive linearly-bounded adversary was studied in~\cite{CK}, wherein  the authors demonstrated a randomized algorithm with expected work  $\cO(m\log^*m)$.

The authors of~\cite{CGKS} developed a
deterministic algorithm with effort (i.e. sum of work and messages sent during the execution) $\cO(n+m^a)$, for a
specific constant~$1<a<2$, against the
unbounded adversary which may crash all but one machine.  They presented the first algorithm for this type of adversary with both work and  communication $o(n+m^2)$, \jm{where communication is understood as the total
number of point-to-point messages sent during the execution}. 
They also gave an algorithm achieving both work and communication
$\cO(n+m\log^2 m)$ against a strongly-adaptive linearly-bounded adversary.

In~\cite{GKS} an algorithm  based on gossiping protocol 
with work $\cO(n+m^{1+\varepsilon})$, for any fixed
constant~$\varepsilon$, is demonstrated.  
In~\cite{KS03}  \DA\ is studied for an asynchronous 
message-passing mode when executions are restricted such that 
every message delay is at most~$d$. The authors proved  $\Omega(n+md\log_d m)$ on the expected work. 
They developed several algorithms, among them a deterministic one with 
work $\cO((n+md)\log m)$. 
Further developments and a more comprehensive related literature can be found in the book by  Georgiou and Shvartsman~\cite{GSbook}.

The work closest to ours is the \DA\ problem of unit length jobs on a multiple access channel with no collision detection, 
which was first studied in~\cite{CKL}.
The authors have showed a lower bound of $\Omega(n+m\sqrt{n})$ on the work when there is no crash,
and $\Omega(n+m\sqrt{n}+m\min\{f,n\})$
in the presence of crashes against an adaptive $f$-bounded adversary.
They also proposed an algorithm, called \textsc{Two-Lists}, with optimal performance $\Theta(n+m\sqrt{n}+m\min\{f,n\})$.
A recent paper~\cite{KKM2017} discusses  different models of adversaries on a multiple access channel in the context of performing unit length jobs.

Extending the study of unit length jobs to arbitrary length jobs has taken place in other context, c.f.,~\cite{KWZ2017} for recent results and references to such
extension in the context of fault-tolerant centralized scheduling.
There has been a long history of studying preemption vs non-preemption, e.g.,~\cite{McN59,LST16, CG91}, to which our work also contributes.

\pw{The problem we study also finds application in other areas, e.g., in window scheduling.
It is common in applications (like broadcast systems~\cite{Bar-NoyNS03,GJW96}) that requests need to be served periodically in the form of windows.
Our problem could be applied to help selecting a window \dk{length} based on controlling work (\dk{aka} the number of available machine steps).}

\subsection{Our results}

In this paper, we consider \jm{fault-tolerant} scheduling of~$n$ arbitrary length jobs on~$m$ identical machines reliably over a multiple accessed channel.
\pw{Our contributions are threefold, on: deterministic preemptive setting, deterministic non-preemptive setting, and
randomized preemptive setting.}
In the setting with job preemption, we prove a lower bound 
$ L + \Omega(m\sqrt{L} + m\min\{f, L\} + m\alpha)$ 
on work, where $L$ is the sum of lengths of all jobs and $\alpha$
is the maximum length, which holds for both deterministic 
and randomized algorithms against an adaptive \fba\ adversary.
We design a corresponding deterministic distributed algorithm, 
called \algx, achieving work
$ L + \cO(m\sqrt{L} + m\min\{f, L\} + m\alpha)$, which matches the lower
bound asymptotically with respect to the overhead.
In the model without job preemption, we show slightly higher lower bound
$ L + \Omega\left(\frac{L}{n}m\sqrt{n} + \frac{L}{n} m\min\{f, n\} + m\alpha \right)$ on work, implying a separation between the two settings with and without preemption.
Similarly as in the previous setting, it holds for both deterministic 
and randomized algorithms against an adaptive \fba\ adversary, thus showing
that randomization does not help against adaptive adversary regardless of
(non-)preemption setting.
We develop a corresponding deterministic distributed algorithm, called \algy,
achieving work $ L + \cO(\frac{L}{n}m\sqrt{n} + \alpha m \min\{f, n\})$.
We then investigate the scenario with randomization
and derive a randomized distributed preemptive algorithm, called \algz, against non-adaptive adversaries.
This \jm{is a Las Vegas} algorithm \jm{that} achieves \jm{expected} work $\cO(L+m\sqrt{L}+m\alpha)$, and thus shows that
randomization helps against the non-adaptive adversary, lowering the performance
to the lower bound $ L + \Omega(m\sqrt{L} + m\alpha)$ (see Lemma \ref{lem22}).
The results are discussed and compared in details in Section~\ref{sec51}.
Figure~\ref{fig4} illustrates the complexity for different variants of the considered problem.

\begin{figure}[tb]
 \begin{center}
 \includegraphics[scale=0.7]{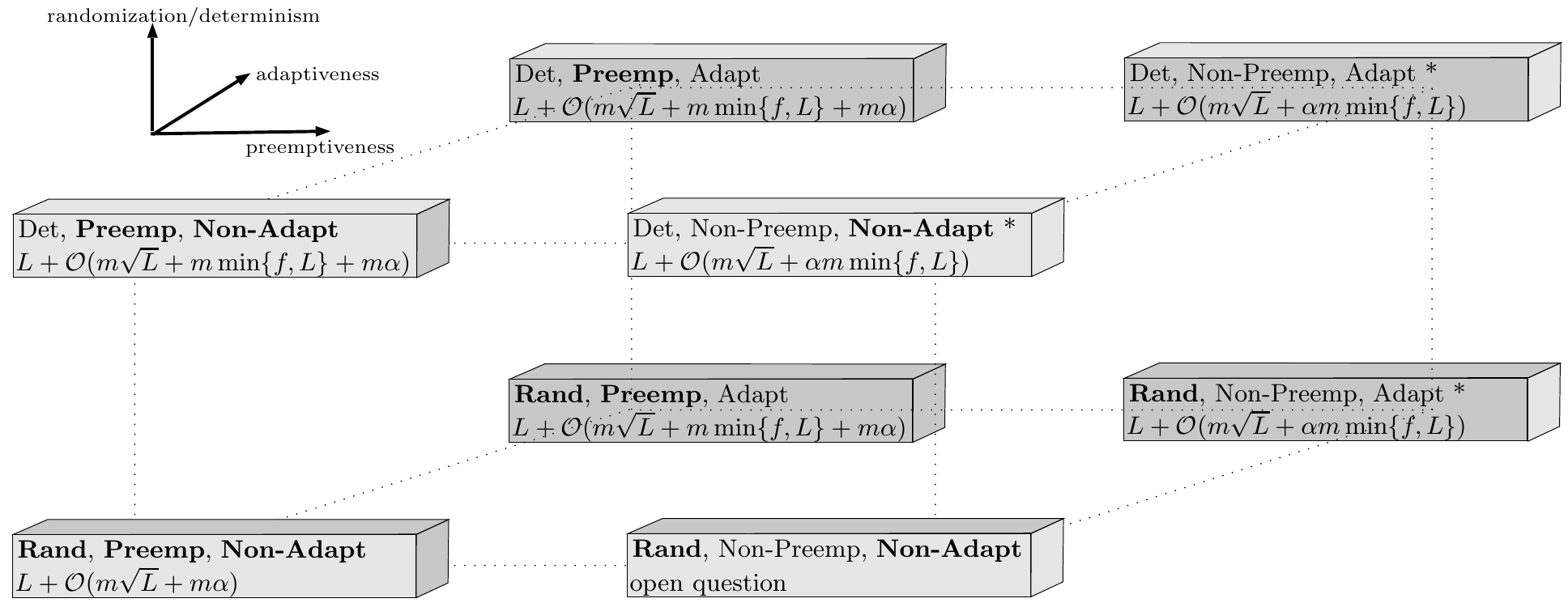}
 \caption{Illustration of the complexity of scheduling on MAC.
Legend: Rand - randomized algorithm, Det - deterministic algorithm, 
 Preemp - preemptive setting, Non-Preemp - non-preemptive setting, Adapt - adaptive adversary, Non-Adapt - non-adaptive adversary.
 \dk{We are considering three different features depending on whether they are present or not: preemptiveness feature changes
 	along the X-axis, adaptiveness of the adversary changes along
 	the Y-axis, and randomization/determinism changes along
 	the Z-axis.
 The features written in bold denote an easier configuration - for instance 
 (\textbf{Rand}, Non-Preemp, \textbf{Non-Adapt}) denotes a configuration with randomization, against a Non-Adaptive adversary, yet tasks are without preemption.
 Boxes with * indicate that the presented solution for such configuration has a gap when confronted with the lower bounds that we present.
 ``Open'' stands for an open question.}
 \pw{The results on the front facet \dk{(lighter boxes)} are for non-adaptive adversary indicating that
(i) preemption improves deterministic algorithms and (ii) randomization improves preemptive algorithms.
The results on the rear facet \dk{(darker boxes)} are for adaptive adversary indicating that preemption improves both deterministic and randomized algorithms.}} 
 \label{fig4}
 \end{center}
\end{figure}

\subsection{\pw{Paper} structure}

Section~\ref{sec11} describes the model and technical preliminaries.
In Section~\ref{sec21} we prove a lower bound and provide an algorithm matching its performance formula (asymptotically with respect to an overhead) in the model with preemption.
Section~\ref{sec31} provides similar study for the model without preemption.
Section~\ref{sec41} is dedicated to a randomized algorithm solving the problem with minimal work complexity.
The results in various settings are discussed and compared in Section~\ref{sec51},
and final conclusions given in Section~\ref{sec61}.

\section{Technical preliminaries}
\label{sec11}

In this section we describe the formal model, i.e.: machine capabilities, communication model, failure model and the concept of adversary,
the complexity measure used for benchmarking our algorithms and a high-level specification of a black-box procedure \tpbb\ that we will be using while designing our algorithms.
Additionally we present definitions of performing jobs and the technicalities regarding preemption.

\subsection{Machines}

In our model we assume having $ m $ \pw{identical} machines, with unique identifiers from the set $ \{1,\dots, m\} $. The distributed system of those machines is synchronized with a global clock,
and time is divided into synchronous time slots, called \textit{rounds}. All the machines start simultaneously at a certain moment. Furthermore every machine may halt voluntarily.
A halted machine does not contribute to work, in contrast to an operational machine that contributes to the work performance measure. 
Precisely, every operational machine generates a unit of work per round \pw{(even when it is idle)}, that has to be included while computing
the complexity of an algorithm.
In this paper by $ M $ we denote the number of operational,
i.e., not crashed, machines. \jm{$ M $ may change during the execution.}

\subsection{Communication model}

The communication channel for machines is the, widely considered in literature, {\em multiple-access channel} \cite{Chl, Gal}, also called a {\em shared channel}, where a broadcasted message reaches 
every operational machine. We do not allow simultaneous successful transmissions of several messages. In what follows our solutions work on a channel without collision detection, hence
when more than one message is transmitted at a certain round, then the machines hear a signal 
indifferent 
from the background noise.
\pw{A job is said to be \emph{confirmed} if a machine completing the job successfully communicates this information via the shared channel.}

\subsection{Adversarial model}

Machines may fail 
\dk{by crashing,}
which happens because of an adversary. One of the factors that describes an adversary is the power~$ f $, that is, the total number of failures that may be enforced.
We assume that $ 0 \leq f \leq m-1 $, thus at least one machine remains operational until an algorithm terminates. Machines that were crashed neither restart nor contribute to work.

We consider two adversarial models.
An \emph{adaptive $f$-bounded adversary} 
can observe the execution and make decisions about up to $f$
crashes online.
%
A \emph{non-adaptive $f$-bounded adversary},
\jm{in addition to} choosing the faulty subset of $ f $ machines, the adversary has to declare in the beginning of the execution in which rounds
those crashes will occur, i.e., for every machine declared as faulty, there must be a corresponding \jm{round number} in which the fault will take place.
It is worth noticing that in the context of deterministic algorithms such an adversary is consistent with the adaptive adversary that may decide online which
machines will be crashed at any time, as the algorithm may be simulated by the adversary before its execution, providing knowledge about the best decisions to be taken.
Finally, an $(m-1)$-bounded adversary is also called an {\em unbounded} adversary.

\subsection{Complexity measure}

The complexity measure to be used in our analysis is, as mentioned before, \textit{work}, also called {\em available machine steps}. 
It is the number of available machine steps for computations. This means that each operational machine
that did not halt contributes a unit of work in each round even if it is idling. 

Precisely,
assume that an execution starts when all the machines
begin simultaneously in some fixed round $ r_{0} $. Let $ r_{v} $ be the round when machine $ v $ halts (or is crashed). Then its work contribution is equal to $ r_{v} - r_{0} $. 
Consequently the algorithm complexity is the sum of such expressions over all machines i.e.: $ \sum_{1 \leq v \leq m}(r_{v} - r_{0}) $.

\subsection{Jobs and reliability}

\jm{We assume that the list of jobs is known to all the machines} and we expect that machines perform all $ n $ jobs as a result of executing an algorithm. We assume that \textit{jobs have arbitrary lengths}, are
\textit{independent} (may be performed in any order) and \textit{idempotent} (may be performed many times, even concurrently by different machines).

Furthermore a \textit{reliable} algorithm satisfies the following conditions in any execution: 
\begin{enumerate}
 \item \textit{All the jobs are eventually performed, if at least one machine remains non-faulty}.
 \item \textit{Each machine eventually halts, unless it has crashed}.
\end{enumerate}

We assume that jobs have some minimal (atomic or unit) length. For simplicity, we can also think that 
all the lengths are a multiple of the minimal length. As the model that we consider is synchronous, this minimal length may be justified by the round duration required for local computations for each machine.
By $ \ell_{a} $ we denote the length of job~$ a $.
We also use~$ L $ to denote the sum of the lengths of all jobs, i.e., $ L = \sum_{i} \ell_{i} $. Finally, by $ \alpha $ we denote the length of the longest job.

\subsection{Preemptive vs Non-Preemptive Model}

By the means of \textit{preemption} we define the possibility of performing jobs in several pieces. Precisely, consider a job~$ a $, of length $ \ell_{a} $ 
(for simplicity we assume that $ \ell_{a} $ is even) and machine~$ v $ is scheduled to perform job~$ a $ at some time of the algorithm execution. Assume that~$ v $ performs $ \ell_{a}/2 $ 
units of job $ a $ and then reports such progress. When preemption is available the remaining $ \ell_{a}/2 $ units of job~$ a $ may be performed by some machine~$ w $ where $ w \neq v $.

A job $ a $ of length $ \ell_{a} $
means that job~$ a $ consists of $ \ell_{a} $ time units needed to perform it fully. Such view allows to think that the job
is a chain of $ \ell_{a} $ \textit{tasks}, and hence all jobs form a set of chains of unit length tasks.

We further denote by~$ a_{k} $ task~$ k $ of job~$ a $. 
However, when we refer to a single job, disregarding its tasks, we refer it simply to job~$ a $.

\noindent We define two types of jobs \pw{regarding how intermediate progress is handled}:

\begin{itemize}
 \item \textbf{Oblivious job} --- it is sufficient to have the knowledge that previous tasks were done, in order to perform \pw{remaining tasks} from the same job. In other words, any information from the in between progress 
 does not have to be announced.
 
 \item \textbf{Non-Oblivious job} --- any intermediate progress needs to be reported through the channel when \pw{interrupting the job to resume later, and possibly  pass the job to another machine}. 
 
\end{itemize}

In the preemptive oblivious model a job may be abandoned by machine~$ v $ on some task~$ k $ without confirming progress up to this point on the channel and then continued from the same task~$ k $ 
 by \pw{any} machine~$ w $. 
 
 As an example of preemptive oblivious model, consider a scenario that there is a job that a shared array of length~$ x $ needs to be erased. If a machine stopped performing this job at some point, another
 one may reclaim that job without the necessity of repetition by simply reading to which point it has been erased.

 For the preemptive non-oblivious model, consider that
 a machine executes Dijkstra's algorithm for finding the shortest path. If it becomes interrupted, then another machine cannot reclaim this job otherwise than by performing the job
 from the beginning, unless intermediate computations have been shared. In other words preemption is available with respect to maintaining information about tasks.

On the contrary to the preemptive model, we also consider the model without preemption i.e. where each job can be performed by a machine \pw{only} in one piece - when a machine is crashed while
performing such job, the whole progress is lost, even if it was reported on the channel before the crash took place.

\subsection{The Task-Performing Black Box (\tpbb)}
\label{sec:tpbb}

\dk{The algorithms we design in this paper employ a black-box procedure 
	for arbitrary length jobs that is able to perform reliably a subset of input
(in the form of jobs consisting of chains of consecutive tasks) or report that something went wrong.
In what follows, we specify this procedure, called 
{\em Task-Performing Black Box} (\tpbb\ for short), and argue that it can be implemented.
}

\jm{We use the procedure in both deterministic and randomized solutions. Precisely all our algorithms use \tpbb\ and despite of that they achieve different performances in the sense of complexity.
In what follows the most important ideas of our results lie within how to preprocess the input rather than how to actually perform the jobs, so we believe that employing such a black-box improves 
the clarity of presentation.}

\remove{
When considering jobs that have arbitrary lengths, 

we are rather interested in the matter of how to insert a chunk of input (jobs forming chains of consecutive tasks) to a procedure, 
which as a result 
would perform them reliably or inform that something went wrong. For this reason
and for clarity of presentation, we 
concentrate on the most important ideas regarding 
our algorithms and 
treat the parts where and how the jobs are actually performed as a black box procedure. In what follows, we specify this procedure, called 
{\em Task-Performing Black Box} (\tpbb\ for short), and argue that it can be implemented.

}

\paragraph{General properties of \tpbb.} 

A synchronous system is characterized by time being divided into synchronous slots, that we already called rounds. In what follows, each round is a possibility for machines to transmit.
  The nature of arbitrary length jobs leads however to a concept that the time between consecutive broadcasts needs to be adjusted. Especially when a job is long, for some 
  configurations it may be  better to broadcast the fact of performing the job fully, rather than broadcasting semi-progress multiple times.
  Therefore, we assume that \tpbb\ has a feature of changing the duration between consecutive broadcasts and we will call this parameter a \textit{phase}. 
	Denote the length of a phase by $ \phi $.
	\dk{Unless stated otherwise,} we assume that $ \phi = 1 $, i.e., the duration of a phase and a round is consistent, thus machines may transmit in any round.

\paragraph{Input: $ \texttt{\tpbb}(v, d, \texttt{JOBS}, \texttt{MACHINES}, \phi) $}

\begin{itemize}
 \item \jm{$ v $ represents the id of a machine executing \tpbb.}
 \item \tpbb\ takes a list of machines $\texttt{MACHINES}$ and a list of jobs $\texttt{JOBS}$ as the input, yet from the task perspective, i.e., the jobs are provided as a chain of tasks and are performed according to their ids.
 It may happen that a job is not done fully and \pw{only some} initial segment of tasks forming that job could be performed.
 \item Additionally, the procedure takes an integer value $ d $. It specifies the number of machines that will be used in the procedure for performing the provided pieces of jobs, 
 which is related to the amount of work accrued during a single execution. The procedure works in such a way that each working machine is responsible for performing a number of jobs.
 For clarity we assume that \tpbb\ always uses the initial~$ d $ machines from the list of machines.
 \item We call a single execution an \textit{epoch} and the parameter~$ d $ is called the length of an epoch (i.e., the number of phases that form an epoch).
 
 \item \jm{$\phi$ is the length of a phase i.e. the duration between consecutive broadcasts.}
\end{itemize}
\paragraph{Output:} \pw{what jobs/tasks have been done and which machines have crashed.}

\begin{itemize}
 \item Having explained what the length of an epoch is in \tpbb\, \pw{we then describe} the capability of performing tasks in a single epoch,
 \pw{which is} understood as the maximal number of jobs that may be confirmed in an epoch, when
 it is executed fully without any adversarial distractions.
 \jm{Let us first note, that \tpbb\ allows to confirm $ j $ tasks in some round $ j $. This comes from the fact that if a station worked for $ j $ rounds and was able to perform one task per round, then
it can confirm at most $ j $ tasks when it comes to broadcasting in round $ j $.}
 \jm{Therefore, the capability of performing tasks in an epoch} is at least $ \sum_{j=1}^{d} j $ which is the sum of an arithmetic
\jm{series} with common difference equal $ 1 $ \jm{over all rounds of an epoch}. We refer to this capability as a \textit{triangle}, as if we take a geometric approach and draw lines or boxes of increasing lengths,
then drawing this sequence would form a triangle.
 \item As a result of running a single epoch we have an output information about which tasks were actually done and whether there were any machines identified as crashed, when
machines were communicating progress (machines identify which jobs were done after an epoch by comparing information about performed tasks).
\end{itemize}

\pw{A candidate algorithm to}
serve as \tpbb\, is the \textsc{Two-Lists} algorithm from \cite{CKL} 
and \pw{we refer readers to the details therein}. \jm{Nevertheless, we assume that it may be substituted with an arbitrary algorithm fulfilling the requirements that we stated above.}

\begin{fact} \label{fact11}
 \tpbb\ 
performs at most $ \sum_{j=1}^{d} j $ tasks while generating $ md \phi $ work \dk{in each execution}.
\end{fact}

\begin{proof}
 We have at most~$ m $ machines working for~$ d $ phases, and the length of each phase is set to~$ \phi $ rounds, hence this gives $ md \phi $ units of work in total. 
\end{proof}

\section{Preemptive model}
\label{sec21}

In this section we consider the scheduling problem in the model with preemption, which is intuitively an easier model to tackle.
As mentioned before, the multiple-access channel has no collision detection.
We first show a lower bound for oblivious jobs (Section~\ref{sec:preempt_lb}) and 
then present and analyze our algorithm for non-oblivious jobs with matching performance (Section~\ref{sec:preempt_alg}).
\pw{Notice that a lower bound for oblivious jobs is a stronger lower bound as it also applies for non-oblivious jobs
and similarly an upper bound for non-oblivious jobs is a stronger upper bound.
The matching bounds imply that with preemption, the distinction between oblivious jobs and non-oblivious jobs does not matter.}


\subsection{Lower bound}
\label{sec:preempt_lb}


We first recall the minimal work complexity introduced in~\cite{CKL}.

\begin{lemma}[\cite{CKL}, Lemma 2] \label{lem21}
A reliable, distributed algorithm, possibly randomized, performs work $\Omega(n + m\sqrt{n})$ in an execution in which no failures occur.
\end{lemma}

\pw{Recall that~$ L $ denote the total length of jobs
and~$ \alpha $ denote the length of the longest job.}
As jobs are built from unit length tasks, we can look at this result from the task perspective. 
Precisely, $ L $ can be considered to represent
the number of tasks needed to be performed by the system, then the lower bound for our model translates in a straightforward way. In addition, because component~$ n $ appeared in the formula above with factor~$ 1 $ in front of it, then we may take it out of the asymptotic notation.
Furthermore in our model there is a certain bottleneck dictated by the longest job. After all there may be an execution with one long job, in comparison to others being short. This gives the magnitude
of work in the complexity formula.

\begin{lemma} \label{lem22}
 A reliable, distributed algorithm, possibly randomized, performs work $ L + \Omega(m\sqrt{L} + m\alpha)$ in an execution in which no failures occur on any set of oblivious jobs with arbitrary lengths with preemption.
\end{lemma}

\begin{proof}
 $ L + \Omega(m\sqrt{L}) $ follows from Lemma~\ref{lem21}. \pw{It suffices to} concentrate on the remaining part.
 \jm{If $ m $ machines are operational for $ \alpha $ rounds and the adversary does not distract their work, then they generate $ m\alpha $ units of work,
 what is consistent with performing the longest job.}
\end{proof}

In the presence of failures, we extend the following result.

\begin{lemma}[\cite{CKL}, Theorem 2] \label{lem24}
The $f$-bounded adversary, for $ 0 \leq f < m $, can force any reliable, possibly randomized, algorithm 
to perform work 
 $ \Omega(n + m\sqrt{n} + m\min\{f, n\}) $.
\end{lemma}

Combining all the results above, we conclude with the following theorem, setting the lower bound for the considered problem.

\begin{theorem} \label{thm22}
 The adaptive f-bounded adversary, for $ 0 \leq f < m $, can force any reliable, possibly randomized and centralized algorithm 
and a set of oblivious jobs 
 of arbitrary lengths with preemption, to perform work $ L + \Omega(m\sqrt{L} + m\min\{f, L\} + m\alpha)$.
\end{theorem}

\begin{proof}
 The result follows explicitly by combining Lemma~\ref{lem22} and Lemma~\ref{lem24} applied for $ n = L $, as in the preemptive model tasks may be treated as independent jobs.
\end{proof}

\subsection{Algorithm \algx}
\label{sec:preempt_alg}

In this section we present our algorithm Scaling-Triangle MAC Scheduling (\algx\ for short).
Intuitively, we use \tpbb\ which treats the capability of performing jobs as a triangle (see Section~\ref{sec:tpbb}). 
This capability may vary through the execution (because of machine crashes and job completion),
and effectively this is like rescaling the triangle.
The most important assumption for \algx\ is that machines have access to all the jobs and the corresponding tasks that build those jobs.
This means that they know all the jobs (and tasks) and their id's as well as their lengths.

In what follows we assume that each machine maintains three lists: \texttt{MACHINES}, \texttt{TASKS} and \texttt{JOBS}. The first list is a list of operational machines and is updated according to the
information broadcasted through the communication channel. 
If there is information that some machine was recognized as crashed, then this machine is removed from the list. In the context of the \tpbb\ 
procedure, recall that this is realized as the output information. \tpbb\ returns the list of crashed machines so that operational machines may update their lists.

\begin{figure}[htb]
 \begin{center}
 \includegraphics[scale=0.8]{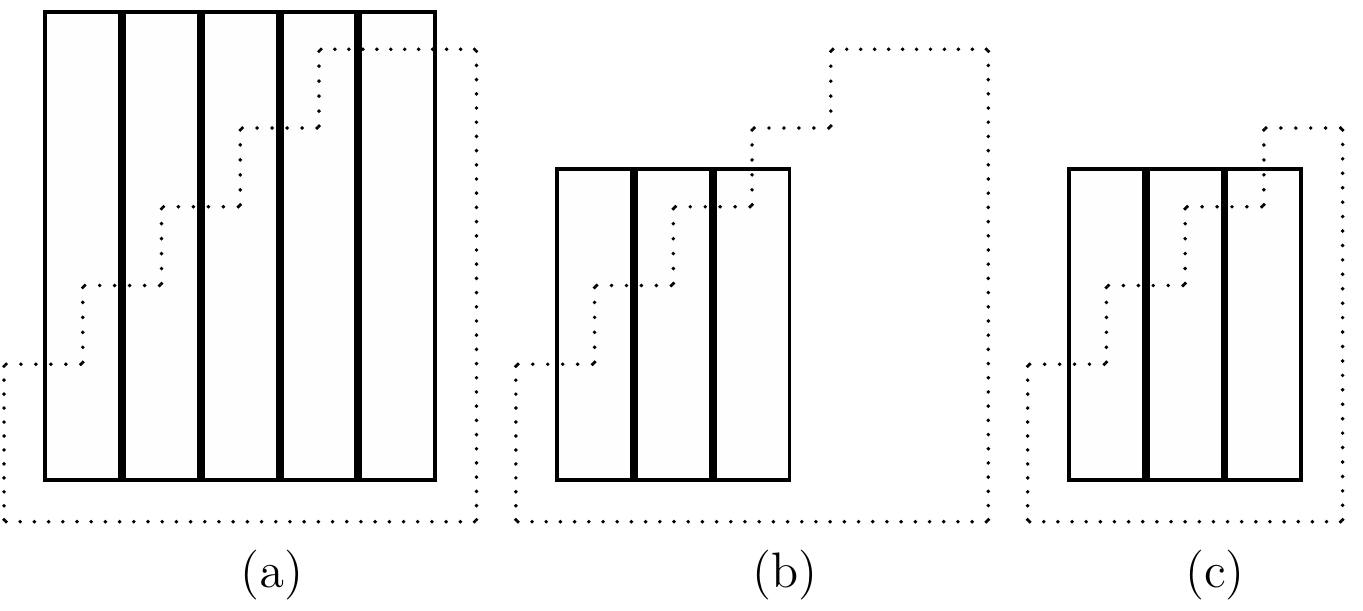}
 \caption{A very general idea about how \algx\ works. The vertical stripes represent jobs and the dotted-line triangle is the capability of the algorithm to perform jobs 
 in a single epoch (or parts of them  i.e. some consecutive tasks). Consequently, if there are enough jobs to pack into an epoch, then it is executed (a). 
 Otherwise, if there are not enough jobs for the current length of an epoch (b), then the length of the epoch is reduced (c).
 This helps preventing excessive idle work.}
 \label{fig0}
 \end{center}
\end{figure}

List \texttt{TASKS} represents all the tasks that are initially computed from the list of jobs. Every task has its unique identifier, what allows to discover \pw{to which job it belongs and its
position in that job}. This allows to preserve consistency and coherency: task~$ k $ cannot be performed before task~$ k - 1 $.
If some tasks are performed then this fact is also updated on the list.

List \texttt{JOBS} represents the set of jobs --- it is a convenient way to know what the consecutive parts of input are for the \tpbb\ procedure. Pieces of information 
on list \texttt{JOBS} may be updated strictly from the ones from list \texttt{TASKS}. \tpbb\ also returns information about which jobs were actually done.
\pw{To ease the discussion}, we denote by $ |\texttt{XYZ}| $ the length of the list XYZ and $ M = |\texttt{MACHINES}| $ the actual number of operational machines.

Jobs are assigned to each machine at the beginning of an epoch (\tpbb\ execution). We assume that machines have instant access to the lengths of jobs. 
While introducing the \tpbb\ procedure we already mentioned the triangle-shaped capacity of the procedure to perform jobs in a single execution. 
Hence, providing \jm{a subset of jobs as the input is consistent with packing them into such a triangle}, \jm{cf., Figure \ref{fig0}}.

We assume~$ d $ is a parameter describing the actual length of an epoch. Initially it is set to $ m $, but it may be reduced while the algorithm is running\jm{, cf., Figure \ref{fig0} (b), (c)}.
In what follows assuming that the length of the epoch $ d $ is set to $ m/2^{i} $ for some $ i $, and we need to fill in a triangle  of size $ \sum_{j=1}^{\frac{m}{2^{i}}} j $,
initially we need to provide $ m/2^{i} $ jobs, that will form the base of the triangle. Jobs are provided as the input in such a way that the shortest ones are preferred and if there are several 
jobs with the same length, then those with lower id's are preferred, cf.,. Figure \ref{fig03}. After having the base filled, it is necessary to look whether there are any gaps in the triangle, 
\jm{see Figure \ref{fig03} (b)}. 
If so, 
another layer of jobs is placed on top of the base layer, and the procedure is repeated\jm{, see Figure \ref{fig03} (c), (d)}. Otherwise, we are done and ready to execute \tpbb.

\begin{figure}[htb]
 \begin{center}
 \includegraphics[scale=0.70]{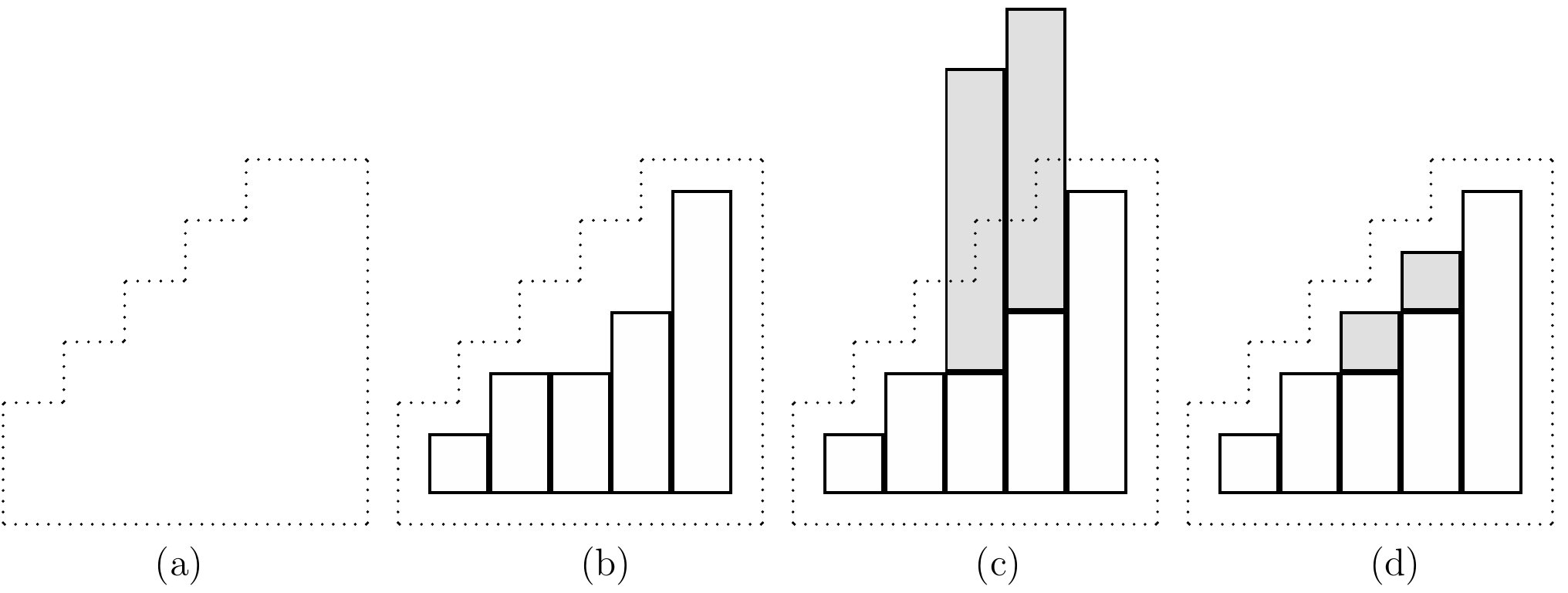}
 \caption{An illustration of job input. (a) Initially there is a certain capability of \tpbb\ (triangle size) for performing jobs. (b) The initial layer is filled with jobs (white blocks), 
 yet there are still some gaps. (c) An additional layer (gray blocks) is placed to fill in the triangle entirely. (d) In fact those additional jobs will only be done partially.}
 \label{fig03}
 \end{center}
\end{figure}

Such an approach allows to ``trim'' longer jobs preferably --- these will have more tasks completed after executing \tpbb\, than shorter ones.
One should observe that performing a transmission is an opportunity to confirm on the channel that the tasks that were assigned to a machine are done. Additionally this confirms that a certain machine is still
operational. In what follows these two types of aggregated pieces of information: which jobs were done, and which machines were crashed, are eventually provided as the output of \tpbb.

When $ d = m/2^{i} $ for some $ i = 1, \dots, \log m $  it may happen that there are not enough tasks (jobs) to fill in a maximal triangle. 
If this happens we will, in some cases, reduce the job-schedule triangle (and simultaneously the length of the epoch) to $ m/2^{i+1} $ and try to fill in a smaller triangle\jm{, cf., Figure \ref{fig0} (b), (c)}. 

Finally, we use $ \tpbb(v, d, \texttt{JOBS}, \texttt{MACHINES}, \phi) $ to denote which machine executes the procedure, the size of the schedule and the length of the epoch,
the list of jobs from which the input will be provided, the list of operational machines, and the phase duration. We emphasize that we described the process of assigning jobs to machines
from the algorithm perspective, i.e., \jm{we illustrated that the system} needs to provide input to \tpbb. Nevertheless, for the sake of clarity we assume that \tpbb\ collects the appropriate
input by itself according to the rules described above, \jm{after having lists \texttt{JOBS} and \texttt{MACHINES} provided as the input.}
\pw{Figures~\ref{fig0} and~\ref{fig03} illustrate the idea of \algx.}

\begin{algorithm}
{
{- initialize $\texttt{MACHINES}$ to a sorted list of all $ m $ names of machines\;}
{- initialize $\texttt{JOBS}$ to a sorted list of all $ n $ names of jobs\;}
{- initialize $\texttt{TASKS} $ to a sorted list of all tasks according to the information from $\texttt{JOBS}$\;}
{- initialize variable $ d $ representing the length of an epoch\;}
{- initialize variable $\phi := 1 $ representing the length of a phase\;}
{- initialize $ i := 0 $\;}
{\Repeat{$|\texttt{JOBS}| > 0$}{
{$ d := m/2^{i}$\;}
{\If{$ |\texttt{JOBS}| \geq  d(d + 1)/2 $}{
{execute $ \tpbb(v, d, \texttt{JOBS}, \texttt{MACHINES}, \phi)$\;}
{update \texttt{JOBS}, \texttt{MACHINES}, \texttt{TASKS} according to the output information from \tpbb\;}
}}
{\Else{ $ i := i + 1 $\;}}
}}
}

\caption{\algx; code for machine $ v $}
\label{alg21}
\end{algorithm}

\jm{The following theorem summarizes the work performance of \algx:}


\begin{theorem} \label{thm21}
 \algx\ performs work $ L + \mathcal{O}(m\sqrt{L} + m \min\{f, L\} + m\alpha)$ against the adaptive adversary 
 and a set of non-oblivious, arbitrary length jobs with preemption.
\end{theorem}

\begin{proof}

\jm{The main part of \algx\ is designed as a repeat loop. Therefore, the total contribution to work of \algx\ is divided between three types of iterations of that loop, which we analyze separately:}

\begin{itemize}
 \item Case 1: \jm{iterations in which} $ |\texttt{JOBS}| \geq d $ and $ |\texttt{TASKS}| \geq  d(d + 1)/2 $
 
 If the algorithm reaches such a condition, it means that there is a significant number of jobs and tasks, thus we may assure that there is neither redundant nor idle work. 
 That is why work accrued in phases with successful transmissions accounts to $ L $ and work in phases with failing transmissions \pw{can be computed as} $ m\min\{f, L\} $. Precisely, there are 
 clearly $ f $ crashes that may occur resulting from the adversary  and because the longest length of an epoch is $ m $, then the  amount of wasted work for each crash is at most $ m $. 
 On the other hand the number of crashes cannot exceed the overall number of tasks, hence the minimum factor.
 
 \item Case 2: \jm{iterations in which} $ |\texttt{JOBS}| < d $ and $ |\texttt{TASKS}| \geq  d(d + 1)/2 $
 
 If we fall into such case, there is a significant number of tasks, but they all form few long jobs. This means that the size of an epoch has to be reduced significantly and those jobs need to be trimmed
 by few initial machines in a certain number of epochs. Thus the magnitude of work for such case is $ \mathcal{O}(m\alpha) $.
 
 \item Case 3: \jm{iterations in which} $|\texttt{TASKS}| <  d(d + 1)/2 $
 
 The algorithm begins with the length of the epoch set as $ m $.
 When the number of tasks is relatively small (no matter what their distribution over jobs is) this means that the algorithm should reduce the size of the schedule triangle to $ m/2 $ in order
 to reduce the number of idling rounds of machines that do not have their jobs assigned. Of course the productive part of work after rescaling the triangle, generated by machines that transmitted successfully is 
 accounted to $ L $ and the wasted part produces $ m\min\{f, L\} $ work.
 
 We now compute the idling work of the remaining machines. The ratio between the initial triangle size and the reduced one is 
 
 $$ \frac{\frac{\frac{m}{2}(\frac{m}{2} + 1)}{2}}{\frac{m(m + 1)}{2}} = \frac{1}{4}\;\frac{m + 2}{m + 1} \geq \frac{1}{4} \ , $$
 hence we conclude that $ \mathcal{O}(1) $ epochs is enough for a single case.

 
 Having in mind that we are computing the overall work in the considered case as a union bound, we examine the inequality $ m(m+1)/2 > L $. It states that there are not enough jobs to fill in 
 the initial triangle --- the size of the triangle is therefore reduced to $ m/2 $. But because its size is reduced to $ m/2 $ it means that now \pw{the number of idling phases within this case is bounded by $ m/2 = \mathcal{O}(\sqrt{L})$.}
 Together with the fact that the number of operational machines is at most $ m $ and that the number 
 of epochs sufficing to perform all the jobs within a single case is constant we compute the total work as follows.
 Following Fact \ref{fact11} we have that $ \mathcal{O}(1) $ epochs $ \times $  $ m $ machines idling at most for $ \mathcal{O}(\sqrt{L}) $ phases of an epoch is the idling
 work accrued in a single case. 
 
 Because there is a logarithmic number of such cases and the reasoning for each of them is consistent with the abovementioned (there are not enough jobs to fill in the triangle of size $ m/2^{i} $, hence
 the size is reduced to $ m/2^{i+1} $), basing once again on Fact \ref{fact11} we compute the total idle work as a union bound of all the cases, having the following:
 
 $$ cm \sum_{i} \frac{\sqrt{L}}{2^{i}} = \mathcal{O}(m\sqrt{L}) \enspace, $$
 where $ c $ is the constant number of epochs required to perform everything within a single case.
 
\end{itemize}

Note that we consider four types of work: idling, wasted, the one following from a bottleneck scenario (when there are few very long jobs) and productive.
Idling work has been bounded by the union of several cases when the length of an epoch is reduced, yet there are still some operational machines that may generate idle work. 

Wasted work is a result of adversary and is clearly related to crashes that may occur. However,
if some task~$ k $ is scheduled to be done in an epoch by machine $ v $, then no other machine apart from $ v $ has task $ k $ scheduled in this particular epoch. In what follows there are two
possibilities: either task~$ k $ is done by~$ v $, and the generated productive work is accounted to $ L $ (only once for task $ k $ in the entire execution),
or machine~$ v $ is crashed and the work accrued by~$ v $ is categorized as wasted (factor $ m \min\{f, L\} $) and $ k $ is scheduled once again in the following epoch.

Consequently, considering all the cases we have that \algx\ has $ L + \mathcal{O}(m\sqrt{L} + m \min\{f, L\} + m\alpha)$ work complexity against an adaptive adversary. 
\end{proof}

As mentioned in the beginning of Section~\ref{sec21},
the lower bound here is proved for oblivious jobs in the preemptive model, while the algorithm works reliably for non-oblivious jobs in the same model,
because \tpbb\ procedure provides any intermediate job performing progress as the output information and all the information is updated sequentially after each epoch.
This implies that the distinction between oblivious and non-oblivious jobs in the preemptive model 
and against 
an adaptive adversary does not matter, thus we finish this section with the following corollary:

\begin{corollary} \label{cor22}
\algx\ is optimal in asymptotic work efficiency for 
jobs with arbitrary lengths with preemption for both oblivious and non-oblivious jobs.
\end{corollary}
\section{Non-preemptive model}
\label{sec31}

In this section we consider a complementary problem of performing jobs when preemption is not available, thus every job needs to be performed in one attempt by a single machine in order to be confirmed.
Again, we will begin with the lower bound and then proceed to the algorithm and its analysis.

\subsection{Lower bound}

The non-preemptive model is characterized by the fact that each job can be performed by a machine in one piece. \jm{This reflects an all-or-nothing policy, i.e.,}
a machine cannot make any intermediate progress while performing a job.
If it starts working on a job then either it must be done entirely or it is abandoned without any tasks \jm{actually} performed. 

In what follows any intermediate broadcasts while performing a job are not helpful for the system. The only meaningful transmissions are those which allow to confirm certain jobs being done. This is the key
reason to consider the notion of a phase, introduced in the model section ---  broadcasts regarding confirming jobs are worth doing only after the entire jobs are done.
We start
the proof of the lower bound by
estimating how many jobs we may be able to confirm in a certain period
based on the average length of a job. Precisely, we will show that the number of jobs with length twice the average does not exceed half of the total number of jobs.

\begin{fact} \label{fact31}
 Let $ n $ be the number of jobs, $ L $ be the sum of all the lengths of jobs and let $ \frac{L}{n} $ represent the average length of a job. 
 Then we have that $ \left|\left\{a: \ell_a > 2\frac{L}{n}\right\}\right| \leq \frac{n}{2}$.
\end{fact}

\begin{proof}
 We have that $ L = \sum_{i \leq n} \ell_i $, so
 $$ L = \sum_{i \leq n} \ell_i \geq \sum_{\left\{a: \ell_i > 2\frac{L}{n}\right\}} \ell_i \geq 2\frac{L}{n} \left| \left\{ a: \ell_a > 2\frac{L}{n}  \right\} \right| \enspace,$$
 hence
 $\frac{n}{2} \geq \left| \left\{ a: \ell_a > 2\frac{L}{n}  \right\} \right| $.
\end{proof}

We now state a general lower bound for the model even if the adversary does not distract the system.

\begin{lemma} \label{lem31}
A reliable algorithm, possibly randomized, performs work $ L + \Omega(\frac{L}{n}m\sqrt{n}) $ in an execution in which no failures occur. 
\end{lemma}

\begin{proof}
It is sufficient to consider the channel with collision detection in which machines could possibly have more information.
Let $ \mathcal{A} $ be a reliable algorithm. We also assume that $ \mathcal{A} $ has the phase parameter set to the average length of the set of jobs.
This is the duration between two consecutive transmissions.

The part $ L $ of the bound follows from the fact that all the jobs are built from $ L $ tasks, thus all the tasks need to be done at least once
in any execution of $ \mathcal{A} $. In other words, there are at least $ L $ units of work needed in order to perform them all.
 
Job $ a $ is \textit{confirmed} at phase $ \gamma $ of an execution of algorithm $ \mathcal{A} $, if either a machine broadcasts successfully and it has performed $ a $ by the end of phase $ \gamma $, or
 at least two machines broadcast simultaneously and all of them, with a possible exception of one machine, have performed job $ a $ by the end of phase $ \gamma $ of the execution.
 All of the machines broadcasting at phase $ \gamma $ and confirming job $ a $ have performed it by then, so at most $ 2\gamma $ jobs can be confirmed at phase $ \gamma $.
 

Let $ \mathcal{E}_{1} $ be an execution of the algorithm when no failures occur. Let~$ v $ be the machine that comes to a halt at some phase $ \gamma' $ in $ \mathcal{E}_{1} $.
 
\noindent
\begin{claim*}
The jobs not confirmed by the end of phase~$ \gamma' $ were performed by $ v $ itself in $ \mathcal{E}_{1} $.
\end{claim*}

\begin{proof}
Suppose, to the contrary, that this is not the case, and let $ b $ be such a job.
Consider an execution, say $ \mathcal{E}_{2} $, obtained by running the algorithm and
crashing  any machine that performed job $ j $ in $ \mathcal{E}_{1} $ just before the machine finishes performing job $ b $ in $ \mathcal{E}_{1} $, and all the remaining machines, except for $ v $,
crashed at phase $ \gamma' $.
The broadcasts on the channel are the same during the first $ \gamma' $ phases in $ \mathcal{E}_{1} $ and $ \mathcal{E}_{2} $.
Hence all the machines perform the same jobs in $ \mathcal{E}_{1} $ and $ \mathcal{E}_{2} $ till the beginning of phase $ \gamma' $.
The definition of $\mathcal{E}_{2} $ is consistent with the power of the \textit{unbounded} adversary.
The algorithm is not reliable because job $ b $ is not performed in 
$ \mathcal{E}_{2} $ and machine $ v $ is operational.
We obtain contradiction that justifies the claim.
\end{proof}

We \pw{compute} the contribution of the machine $ v $ to work.
The total number of jobs confirmed in $ \mathcal{E}_{1} $ is at most 
\[
2(1+2+\ldots+\gamma')=\mathcal{O}((\gamma')^2)\ .
\]
Suppose some $ n' $ jobs have been confirmed by the end of phase $ \gamma' $. Each job contributes $ \frac{L}{n} $ units of work because of the phase parameter.
The remaining $ n-n' $ jobs have been performed by $ v $ and each of them required $ \frac{L}{n} $ units of work to be done as well.
The work of $ v $ is therefore 
\[
\Omega\left(\frac{L}{n}\sqrt{n'}+\frac{L}{n}(n-n')\right)=\Omega\left(\frac{L}{n}\sqrt{n}\right)\ ,
\]
which completes the proof.
\end{proof}

\begin{lemma} \label{lem32}
 A reliable, distributed algorithm, possibly randomized,
performs work 
 $ L + \Omega\left(\frac{L}{n}m\sqrt{n} + \frac{L}{n} m\min\{f, n\} + m\alpha \right)$ 
in an execution
with at most $ f $ failures against an adaptive adversary.
\end{lemma}

\begin{proof}
We consider two cases, depending on which term dominates the bound.
If it is $ L + \Omega\left(\frac{L}{n} m\sqrt{n} + m\alpha\right)$, then the bound
follows from combining Lemma~\ref{lem31} with Lemma~\ref{lem22} --- if the longest job was a bottleneck in the preemptive model, then certainly it is a bottleneck in the non-preemptive model.

Consider the case  when $\Omega(\frac{L}{n} m \min \{ f,n \})$ determines the magnitude of the bound.
Denote $g=\min\{f,n\}$. As for now, we may assume that the duration of a phase, i.e., the duration between consecutive broadcasts, is set to some value $ \phi $ and that it may even
change multiple times during any execution.
Let $\cE_1$ be an execution obtained by running the algorithm
and crashing any machine that wants to broadcast as a single one
during the first $g/4$ phases.
Denote as $A$ the set of machines failed in $\cE_1$.
The definition of $\cE_1$ is consistent with the power of the \fba\ adversary, since $|A|\le g/4\le f$.

\begin{claim*}
No machine halts by phase $g/4$ in execution $\cE_1$.
\end{claim*}

\begin{proof}
Suppose, to the contrary, that some machine $v$ halts before phase $g/4$
in $\cE_1$.
We show that the algorithm is not reliable.
To this end we consider another execution, say, $\cE_2$ that
can be made to happen by the \ua\ adversary.
Let $a$ be a job which is performed in $\cE_1$ by at most
\[
\frac{mg}{4(n-g/4)}\le \frac{mg}{3n}
\]
machines, except for machine~$v$, during the first $g/4$ phases.
It exists because $g\le n$.
Let~$B$ be this set of machines. 
Size $|B|$ of set~$B$ satisfies the inequality 
\[
|B|\le \frac{mg}{3n}\le \frac{m}{3}\ .
\]
We define operationally a set of machines, denoted $C$, as follows.
Initially $C$ equals $A\cup B$.
Notice that the inequality $|A\cup B|\le 7m/12$ holds.
If there is any machine that wants to broadcast during the first $g/4$
phases in $\cE_1$ as the only machine not in the current $C$,
then it is added to $C$.
At most one machine is added to $C$ for each among the first $g/4\le m/4$
phases of $\cE_1$, so $|C|\le 10m/12 <m$. 

Let~$\cE_2$ be an execution obtained by failing all the machines in $C$
at the start and then running the algorithm.
The definition of $\cE_2$ is consistent with the power of  the \ua\ adversary.
There is no broadcast heard in $\cE_2$ during the first $g/4$ phases.
Therefore each machine operational in $\cE_2$ behaves in exactly
the same way in both $\cE_1$ and $\cE_2$ during the first $g/4$ phases.
Job~$a$ is not performed in execution $\cE_2$ by the end of phase $g/4$,
because the machines in~$B$ have been failed and the remaining ones
behave as in $\cE_1$.

Machine $v$ is not failed in~$\cE_2$ and so it performs the same actions in 
both~$\cE_1$ and~$\cE_2$.
Consider a new execution, denoted $\cE_3$.
This execution is like $\cE_2$ till the beginning of phase $g/4$, then all the machines, except for~$v$, are failed.
The definition of $\cE_3$ is consistent with the power of  the \ua\ adversary.
Machine $v$ is operational but halted and job $a$ is still
outstanding in $\cE_3$ at the end of phase $g/4$.
We conclude that the algorithm is not reliable.
This contradiction completes the proof of the claim.
\end{proof}

Consider the original execution $\cE_1$ again.
It follows from the claim that there are at least $m-g/4=\Omega(m)$ machines
still operational and non-halted by the end of phase~$g/4$ in execution~$\cE_1$.
The duration of a phase is dictated by lengths of jobs. Hence even if $ \phi $ was changed $ n $ times, \jm{phases} lasted $ L/n $ on average.
Thus we may bound the magnitude of $ \phi $ by the initial average length of jobs $ L/n $ and consequently machines have generated work $\Omega\left(\frac{L}{n} \cdot m\cdot \min \{f,n\}\right) $ 
by the end of phase $ g/4 $,
what ends the proof.
\end{proof}

\begin{corollary}
\label{cor33}
 There are non-preemptive job inputs for which a reliable, distributed algorithm, possibly randomized, 
performs work
$ L + \Omega\left(\frac{L}{n}m\sqrt{n} + \alpha m\min\{f, n\} \right)$
against adaptive \fba\ adversary. 
\end{corollary} 
\subsection{Algorithm \algy}

\begin{figure}[thb]
 \begin{center}
 \includegraphics[scale=0.7]{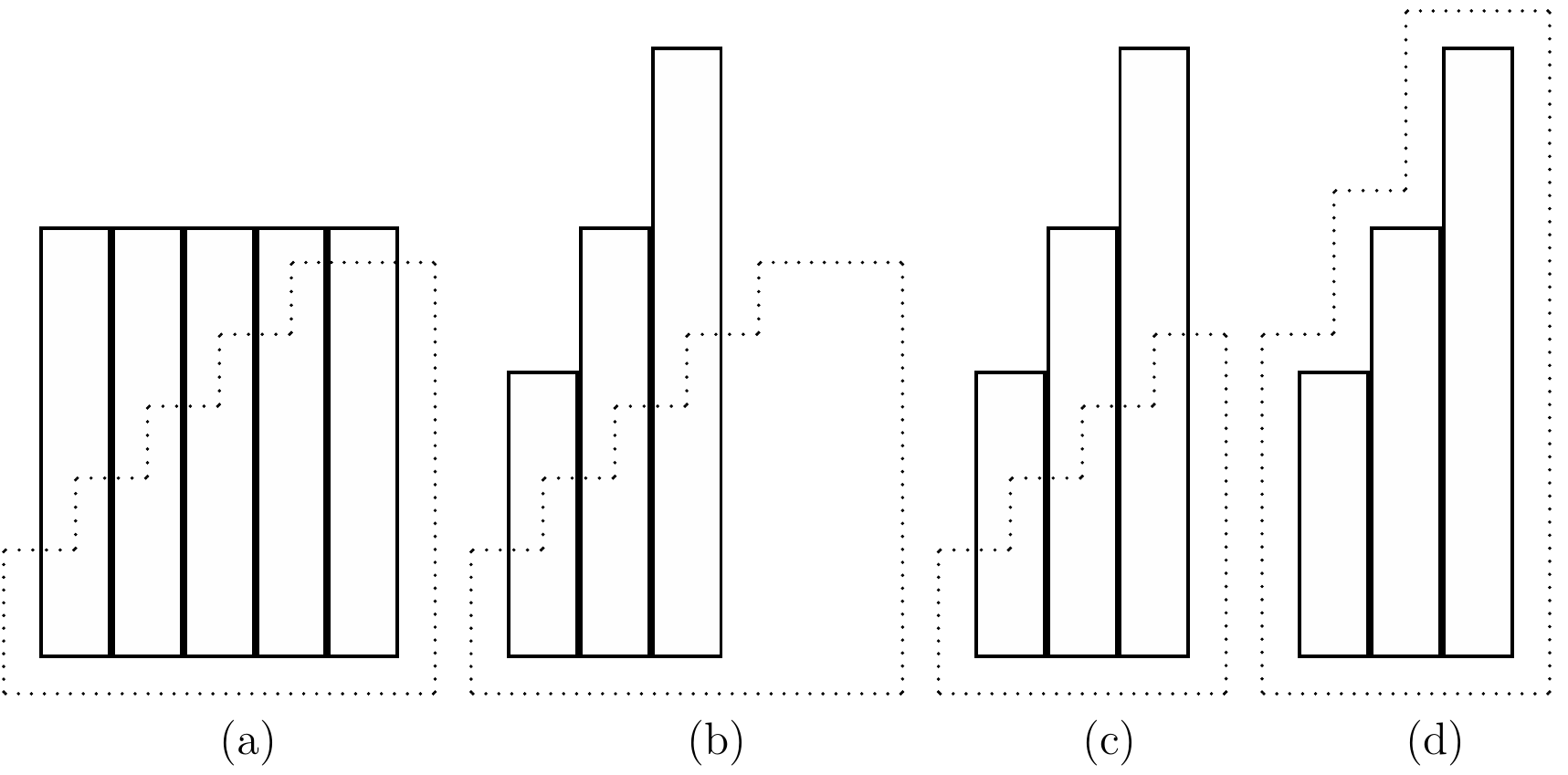}
 \caption{Most important features of \algy.
 Similarly to \algx\, we apply the method of reducing idle work by reducing the input size (epoch length) for the task performing procedure ((a), (b), (c)). 
 Additionally we pump phases (the duration between consecutive broadcasts) in order to be able to fill in jobs entirely - in the model without preemption jobs cannot be done partially (d).}
 \label{fig1}
 \end{center}
\end{figure}

\begin{algorithm}
{
{- initialize $\texttt{MACHINES}$ to a sorted list of all $ m $ names of machines\;}
{- initialize $\texttt{JOBS}$ to a sorted list of all $ n $ names of jobs\;}
{- initialize $\texttt{TASKS} $ to a sorted list of all tasks according to the information from $\texttt{JOBS}$\;}
{- initialize variable $ d $ representing the length of an epoch\;}
{- initialize variable $\phi $ representing the length of a phase\;}
{- initialize $ i := 0 $\;}
{\Repeat{$|\texttt{JOBS}| > 0$}{
{$ d := m/2^{i} $\;}
{\If{$ |\texttt{JOBS}| \geq d(d+1)/2 $}{
{\tcp{set $\phi$ to the current average length of a job}
{$\phi := |\texttt{TASKS}|/|\texttt{JOBS}| $\;}
{execute $ \tpbb(v, d, \texttt{JOBS}, \texttt{MACHINES}, \phi)$\;}
{update \texttt{JOBS}, \texttt{MACHINES}, \texttt{TASKS} according to the output information from \tpbb\;}
}}
{\Else{ 
\If{$|\texttt{TASKS}| < d(d+1)/2 $}{
$ i := i + 1 $\;
}
\Else{
{$\phi := 1 $\;}
{execute $ \tpbb(v, \min\{|\texttt{JOBS}|, |\texttt{MACHINES}|\}, \texttt{JOBS}, \texttt{MACHINES}, \phi)$\;}
{update \texttt{JOBS}, \texttt{MACHINES}, \texttt{TASKS} according to the output information from \tpbb\;}
}
}}
}}
}
}

\caption{\algy; code for machine $ v $}
\label{alg31}
\end{algorithm}

We now design a 
scheduling algorithm 
in the non-preemptive model of jobs. 
Its work complexity is $ L + \mathcal{O}(\frac{L}{n}m\sqrt{n} + \alpha m \min\{f, n\})$ against the adaptive \fba\ adversary.
%
At first glance, algorithm Deforming-Triangle MAC Scheduling (\algy\ for short) 
is similar to \algx\ introduced in the previous section and on the top of its design there is the \tpbb\ procedure for performing jobs.
Roughly speaking, the algorithm, repetitively, tries to choose tasks that could be packed into a specific triangle (parameters of which are controlled by the algorithm) and feed them to \tpbb.
Furthermore, once again the main goal of the algorithm is to avoid redundant or idle work. Hence, the main feature is to examine whether there is an appropriate number of jobs in order to fill in an epoch of \tpbb. However, as we are dealing with the non-preemptive model, jobs cannot be done in small pieces (tasks). 

Consequently, apart from checking the number of jobs (and possibly reducing the size of an epoch)
the algorithm also changes the phase parameter, i.e., the duration between consecutive broadcasts --- this is somehow convenient in order to know how the input jobs should be 
fed into \tpbb.
It settles a ``framework'' (epoch) with ``pumped rounds'' that should be filled in appropriately.

Using Fact~\ref{fact31} assures 
that if the size of an epoch is reduced accordingly to the number of jobs and the phase parameter is set accordingly to the actual average length
of the current set of jobs (i.e., $ |\texttt{TASKS}|/\texttt{JOBS}|$ in the code of Algorithm~\ref{alg31}), then we are able to provide the input for \tpbb\ appropriately and expect that 
at least
$ \sum_{j=1}^{\frac{m}{2^{i}}} j $ jobs will be performed for some size parameter $ i $ as a result of executing \tpbb, without yet taking into account crashes.

Finally, there is a special and quite pessimistic case, that makes a difference between \algx\ and \algy. Specifically, if there are not many jobs left, but they are very long, 
then the algorithm needs to switch to a slightly different \tpbb\ phase parameter $ \phi $, mainly $ \phi = 1 $, so that transmissions are enabled in any given round. 
This is in order not to waste too much idle work
because of the jobs being significantly long for the considered case.
\pw{Figure~\ref{fig1} illustrates the idea of \algy.}

\begin{theorem} \label{thm31}
\algy\
performs work $ L + \mathcal{O}(\frac{L}{n}m\sqrt{n} + \alpha m \min\{f, n\})$ against the adaptive \fba\ adversary. 
\end{theorem}

\begin{proof}
We analyze similar cases as in the proof of Theorem~\ref{thm21} w.r.t. \algx,
using the same methodology.

\begin{itemize}
 \item Case 1: \jm{iterations in which} $ |\texttt{JOBS}| \geq d $ and $ |\texttt{TASKS}| \geq  d(d + 1)/2 $.
 
 For this case we have a great number of jobs to perform and the only subtlety is to take care of the appropriate phase length, that is adjusted before each epoch. Similarly as in Theorem~\ref{thm21},
 productive work resulting from successful job performance accrues $ L $ units of work. If a machine crashes, its work in the epoch is at most $ \phi m $, which is equal to the maximal 
 length of the epoch rescaled by $ \phi $ that is the phase parameter. Because we have $ f $ possible crashes and $ \phi $ is always adjusted to the current average, 
 then we may \jm{upper bound} this kind of work as $\frac{L}{n} m\min\{f, n\} $, taking into consideration that the number of jobs $ n $ may be less than the adversary's 
 possibilities of crashes.

\item Case 2: \jm{iterations in which} $ |\texttt{JOBS}| < d $ and $ |\texttt{TASKS}| \geq  d(d + 1)/2 $.
 
In this case there are not too many jobs, but they are long. Consequently, \algy\ changes the phase parameter of \tpbb\ to $ 1 $ so that transmissions are available at any given time,
 hence the schedule of transmissions is fixed according to list \texttt{MACHINES} and \texttt{JOBS}. In this case each machine has barely one job to take care of and machines with shorter 
 jobs and/or lower job ID's have the opportunity to confirm their jobs before machines having longer jobs scheduled. Obviously the only transmissions that make sense are those that confirm performing
 a particular job.
 
 The length of an epoch is now set to the number of jobs, unless the number of jobs is greater than the number of currently operational machines --- for such case the length of the epoch is 
 simply $ M \le m$. In this case several epochs may be needed to perform all the jobs.
 
 The execution for such case is likely to be consistent with the scenario where machines will work silently for some time (long jobs) and then start communicating progress one after another.
 Hence if now a machine performs a successful transmission, it means that it performed a job and the work accrued by this fact accounts to factor $ L $. Otherwise, if there was a crash caused by the
 adversary then the maximum number of wasted work units accrued is $ \alpha $. While there are as many as $ f $ crash possibilities, we have that $ \alpha f $ is upper bound on the amount of wasted work.
 
 Finally, if there is a number of machines idling because some others are performing few long jobs, then this process generates the amount of $ \alpha f $ work multiplied by $ m $. 
%
We therefore conclude that the overall amount of work accrued for the considered case is $ m\alpha \min\{f, n\} $, given the sub-case that the number of jobs is less than the number of crashes.
 
\item Case 3: \jm{iterations in which} $|\texttt{TASKS}| < d(d+1)/2 $.
 
When there are not enough jobs to provide the input to \tpbb, then the size of a single epoch of \tpbb\ is reduced in order to prevent excessive idle or redundant work. Obviously, the productive and wasted
 work is $ L $ and $ \frac{L}{n} m\min\{f, n\} $ respectively, 
\dk{and it remains to bound the amount of idle work.}

Likewise in the proof of Theorem~\ref{thm21}, we use a union bound over a logarithmic number of cases for the idle kind of work.
Assume that the length of the epoch is set to $\frac{m}{2^{i}} $, for some $ i $, and there are not enough tasks to fully fill the input of \tpbb, hence the length of the epoch is reduced with parameter $ i+1 $.
The ratio between both of these cases is 
 
 $$ \frac{\frac{\frac{m}{2^{i+1}}(\frac{m}{2^{i+1}} + 1)}{2}}{\frac{\frac{m}{2^{i}}(\frac{m}{2^{i}} + 1)}{2}} = \frac{1}{4}\;\frac{m + 2^{i+1}}{m + 2^{i}} \geq \frac{1}{4} \ , $$
 thus we conclude that the reduced epoch is at most $ 4 $ times smaller than the initial one. In what follows, the number of reduced length epochs sufficing to perform all the tasks is constant.
 What is more, the phase parameter will be set appropriately to the average length of the current set of tasks, hence, according to Fact \ref{fact31}, this will be a correct setting to fill in the input for \tpbb.
 
 Taking these facts altogether and a similar reasoning as in Theorem~\ref{thm21} over the logarithmic number of cases, we have that 
 $$ cm \sum_{i} \frac{\sqrt{n}}{2^{i}} = \mathcal{O}(m\sqrt{n}) $$
 phases of work are needed, where $ c $ is the constant number of epochs required for performing all the jobs. Once again: we have a sum over a logarithmic number of cases characterized by the epoch length, where
 at most $ m $ machines are operational, and they are working for $ c $ executions of \tpbb, hence the estimate above.
 
 The final question is however how we can bound the duration of all the phases from above for the idle work estimation.

 Once again, we calculate a union bound over a logarithmic number of cases. \jm{We use $ L_{i} $ for denoting the number of outstanding jobs in consecutive cases.}
 The initial arithmetic average length of a job is $ L_{1}/n $, where $ L_{1} = L $. It is set for the time of performing at most $ n/2 $ jobs, 
 because if we fell into the condition that $|\texttt{TASKS}| < d(d+1)/2 $,
 then at most two epochs suffice to perform all the jobs. And because we change the average length of a phase after each epoch, hence we argued that the initial average is set for $ n/2 $ jobs.
 
 The next step is that we fall into the case when an average of $ L_{2}/(n/2) $ is set for at most $ n/4 $ jobs, and other cases follow similarly. $ $
 The key observation is that the consecutive average phase durations are set for monotonically decreasing periods, where each one is set for at most half the ``time'' than the preceding one.
 Of course $ L = L_{1} \geq L_{2} \geq \dots \geq L_{\log n} $. In what follows we have that the initial average length
 of a job dominates the weighted average of all the averages that occur over all the cases:
 
 $$ \frac{L}{n} \geq \frac{L/2 + L/4 + \dots + L/2^{\log n}}{n} \geq \frac{L_{1}/2 + L_{2}/4 + \dots + L_{\log n}/2^{\log n}}{n} \enspace. $$
 We therefore showed that the idle work factor equals $\mathcal{O}\left(\frac{L}{n}m\sqrt{n}\right)$.
\end{itemize}

As in Theorem~\ref{thm21}, the analysis of \algy\ consists of idle, wasted and productive work. Likewise, if a job is performed and it has been confirmed, the work accounts to $ L $. Otherwise,
it accounts to $ \alpha m\min\{f, L\} $. In what follows, factor $ L $ 
appears in the complexity formula with constant $ 1 $, similarly as it was in \algx\ analysis.
We conclude that \algy\ has $ L + \mathcal{O}(\frac{L}{n}m\sqrt{n} + \alpha m \min\{ f, L\})$ work complexity.
%
\end{proof}

\remove{
\noindent We finish this section with a general corollary about the optimality of \algy\ for the considered problem.

\begin{corollary} \label{cor31}
 Algorithm Y is optimal in asymptotic work efficiency for the channel without collision detection and jobs with different lengths without preemption.
\end{corollary}
}
\section{Towards randomized solutions} \label{sec41}

We \pw{have already showed} that preemption is helpful when considering arbitrary lengths of jobs that need to be done by a distributed system. However, a natural question appears whether we may 
perform even better when randomization is available for the considered problem and its corresponding model. The answer seemed to be positive for the non-adaptive adversary,
hence we investigate this issue by \jm{adapting} an algorithm from \cite{KKM2017} and analyzing it for the configuration focusing on jobs with preemption, 
providing an improved algorithm, intuitions about it and a detailed analysis.

\subsection{Algorithm \algz}

\begin{algorithm}
\If{$ m^{2} \leq L $}{execute \algx\;}
\Else{
{initialize $ \texttt{MACHINES} $ to a sorted list of all $ m $ names of machines\;}
{initialize variable $\phi := 1 $ representing the length of a phase\;}
\If{$\log_{2}(m) > e^{\frac{\sqrt{L}}{32}}$}
{execute a $ L $-phase epoch where every machine has all tasks assigned (without transmissions)\;
execute \textsc{Confirm-Work}\;}
\Else{

{$ i = 0 $\;}
\Repeat{$ i = \lceil \log_{2}(m) \rceil $}{

\If{($ \frac{m}{2^{i}} \leq \sqrt{L} $)}
{Execute \algx\;}
\If{\textsc{Mix-And-Test($ i, L, m $)}}{
\Repeat{less than $ \frac{1}{4}\sqrt{L} $ broadcasts are heard}{
Execute $ \tpbb(v, \sqrt{T}, \texttt{TASKS}, \texttt{STATIONS}, \phi)$\;}
}
{increment $ i $ by $ 1 $\;}
}
}
}
\caption{\algz, code for machine $ v $}
\label{algorithm42}
\end{algorithm}

\begin{algorithm}
\KwIn{$i, L, m$}
{$\texttt{coin} := \frac{m}{2^{i}}$\;}
{\For{$\sqrt{L} + \log m$ times}
{
{\If{$v$ has not been moved to the front of list $ \texttt{MACHINES} $ yet}{toss a coin with the probability $ \texttt{coin}^{-1} $ of heads to come up}}
{\If{heads came up in the previous step}{broadcast $ v $ via the channel and attempt to receive a message}}
{\If{some machine $ w $ was heard}{
{move machine $ w $ to the front of list $ \texttt{MACHINES} $\;}
{decrement $ \texttt{coin}$ by $ 1 $\;}
}}
}}
{\If{at least $ \sqrt{L} $ broadcasts were heard}{return \texttt{true}\;}
\Else{return \texttt{false}\;}
}

\caption{\mat, code for machine $ v $ }
\label{algorithm41}
\end{algorithm}

\begin{algorithm}
{$i := 0$ \;}
\Repeat{a broadcast was heard}{

{$\texttt{coin} := \frac{m}{2^{i}}$ \;}
{toss a coin with the probability $ \texttt{coin}^{-1} $ of heads to come up\;}
\If{heads came up in the previous step}{broadcast $ v $ via the channel and attempt to receive a message\;}
\If{some machine $ w $ was heard}{
clear list $\texttt{JOBS}$\;
break\;
}
\Else{
  increment $ i $ by $ 1 $ \;
  \If{ $ i = \lceil \log_{2}(m) \rceil + 1 $ }
  {$ i := 0 $\;}
}
}

\caption{\conwor, code for machine $ v $}
\label{algorithm43}
\end{algorithm}

The name of \algz\ follows from \algx\ and that in fact the actual job performing remains the same. What distinguishes both algorithms is the use of randomization.
Precisely, \algz\ changes the order of machines on list $ \texttt{MACHINES} $ by moving machines that performed a succesful broadcast to front of that list.
This has dual significance - on one hand it prevents flexible adversarial crashes but on the other it also allows to predict to which interval $n \in(\frac{m}{2^{i}}, \frac{m}{2^{i-1}}] $ 
for $ i = 1, \cdots, \lceil\log_{2}(m)\rceil $ does the current number of operational machines belong what helps us with the analysis.

Machines moved to front of list \texttt{MACHINES} are called \textit{leaders}. Because they are chosen randomly, we expect that they will be uniformly distributed in the
adversary order of crashes. This allows us to assume that a crash of a \textit{leader} is likely to be preceeded by several crashes of other machines.

Consider \mat\ at first. If $ p $ is the previously predicted number of operational machines, then each of the machines tosses a coin with the probability of success equal
$ 1/p $. When only one machine did successfully broadcast it is moved to front of list $ \texttt{MACHINES} $ and the procedure starts again with a decreased parameter. 
Otherwise, when none or more than one of the machines broadcasts then silence is heard on the channel, as there is no collision detection. Machines that have already been moved to front do not take part
in the following iterations of the procedure.

\tpbb\ together with \algx\ are used as subprocedures for \algz. The length parameter for \tpbb\ is now set to $ \sqrt{L} $, as this allows us to be sure 
that the total work accrued in each epoch does not exceed $ m\sqrt{L}\phi $ and simultaneously it is enough to perform all the tasks in a constant number of epochs.

\jm{Let us recall that} before the execution the \textit{Non-Adaptive} adversary has to choose $ f $ machines prone to crashes and declare in which rounds they will be crashed what is in fact consistent with declaring 
in what order those crashes may happen. As a result when some of the \textit{leaders} are crashed, we may be approaching the case when $ M \leq \sqrt{L} $ 
and it is possible to execute \algx\ that complexity is $ L + m\alpha $ for such parameters. Otherwise the adversary spends the majority of his possible crashes and the machines may finish 
all the jobs without being distracted.

The design of \algz\ is quite simple: if some specific conditions are not satisfied (see Algorithm \ref{algorithm42} lines 1-10), 
then \textsc{Mix-And-Text} is executed which purpose is to select a number of leaders. Because they are working for $ \sqrt{L} $ rounds then, if no crashes occur, all the jobs will be done 
in a constant number of epochs ($ \sqrt{L} $ machines perform $ 1 + \dots + \sqrt{L} = \mathcal{O}(L) $ tasks in a single epoch). In case of crashes occurring we may need to repeat
\mat\ and draw a new set of leaders. Whenever $ M < \sqrt{L} $ is satisfied, \algx\ is executed that work complexity is $ \mathcal{O}(L + m\alpha) $ for such parameters.

%
%
%
%
%
%
%
%
%

We begin the analysis od \algz\ with the fact stating that whenever $ \frac{m}{2^{i}}\leq \sqrt{L} $ holds, \algx\ is executed, with $ \mathcal{O}(L + m\alpha) $
work.

\begin{fact}
 Let $ m $ be the number of operational machines, and $ L $ be the number of outstanding tasks. Then for $ m \leq \sqrt{L} $ \algx\ work complexity is $ \mathcal{O}(L + m\alpha) $.
\label{fact43}
 \end{fact}

\begin{proof}
If $ m \leq \sqrt{L} $, then the outstanding number of crashes is $ f < m $, hence $ f < \sqrt{L} $.
\algx\ has $L + \mathcal{O}(m\sqrt{L} + m\:\min\{f, L\} + m\alpha) $ work complexity.
In what follows the complexity is $ L + \mathcal{O}(\sqrt{L}\sqrt{L} + \sqrt{L}\:\min\{\sqrt{L}, L\} + m\alpha) $
$ = \mathcal{O}(L + m\alpha) $.
\end{proof}

%

\begin{lemma}
Assume that we have $ M $ operational machines at the beginning of an epoch, where $ \sqrt{L} $ were chosen as leaders.
 If the adversary crashes $ M/2 $ machines, then the probability that there were $ 3/4 $ of the overall number of leaders crashed in this group does not exceed
 $ e^{-\frac{1}{8} \sqrt{L}} $.
 \label{lem42}
\end{lemma}

\begin{proof}
 We have $ M $ machines, among which $ \sqrt{L} $ are leaders.
 The adversary crashes $ M/2 $ machines and \pw{the question is how many leaders were in this group.}
 
 The hypergeometric distribution function with parameters $ N $ - number of elements, $ K $ - number of highlighted elements, $ l $ - number of trials, $ k $ - number of successes, is given by:
 
 $$ \mathbb{P}[X = k] = \frac{\binom{K}{k}\binom{N-K}{l-k}}{\binom{N}{l}} \enspace. $$
 
 \noindent
 The following tail bound from \cite{Hoef} tells us, that for any $ t > 0 $ and $ p = \frac{K}{N} $:
 
 $$ \mathbb{P}[X \geq (p + t)l] \leq e^{-2t^{2}l} \enspace. $$
 
 \noindent
 Identifying this with our process we have that $ K = M/2 $, $ N = m $, $ l = \sqrt{L} $ and consequently $ p = 1/2 $. Placing $ t = 1/4 $ we have that
 
 $$ \mathbb{P}\left[X \geq \frac{3}{4}\sqrt{L}\right] \leq e^{-\frac{1}{8}\sqrt{L}} \enspace. $$
\end{proof}

\begin{lemma}
 Assume that the number of operational machines is in $(\frac{m}{2^{i}}, \frac{m}{2^{i-1}}] $ interval. Then procedure
 \mat($i, L, m $) will return \textit{true} with probability $ 1 - e^{-c\;(\sqrt{L} + \log_{2}(m))} $, for some $ 0 < c < 1 $.
 \label{lem44}
\end{lemma}

\begin{proof}

\begin{claim}
 Let the current number of operational machines be in $ (\frac{x}{2}, x] $.
 Then the probability of an event that in a single iteration of \mat\ exactly one machine will broadcast 
 is at least $ \frac{1}{2\sqrt{e}} $ (where the $\texttt{coin}^{-1}$ parameter is $ \frac{1}{x} $).
\end{claim}

\begin{proof}
Consider a scenario where the number of operational machines is in $ (\frac{x}{2}, x] $ for some $ x $. 
If every machine broadcasts with probability of success equal $ 1/x $ then the probability of an event that exactly one machine will transmit is $ (1 - \frac{1}{x})^{x-1} \geq 1/e $. 
Estimating the worst case, when there are $ \frac{x}{2} $ operational machines (and the probability of success remains $ 1/x $) we have that
$$ \frac{1}{2} \left(1 - \frac{1}{x}\right)^{x \cdot \frac{x-2}{2}} \geq \frac{1}{2\sqrt{e}} \enspace. $$
\end{proof}

According to the Claim the probability of an event that in a single round of \mat\ exactly one machine will be heard is $ \frac{1}{2\sqrt{e}} $.

We assume that $ M \in (\frac{m}{2^{i}}, \frac{m}{2^{i-1}}] $. We show that the algorithm confirms appropriate $ i $ with probability $ 1 - e^{-c\;(\sqrt{L}+\log_{2}m)} $. 
For this purpose we need $ \sqrt{L} $ transmissions to be heard. 

Let $ X $ be a random variable such that $ X = X_{1} + \cdots + X_{\sqrt{L}+\log_{2}(m)}, $ 
\\where $ X_{1}, \cdots, X_{\sqrt{L}+\log_{2}(m)} $ are Poisson trials and

$$
X_{k} = \left\{ \begin{array}{ll}
1 & \textrm{if machine broadcasted,}\\
0 & \textrm{otherwise.}
\end{array} \right 
.$$
We know that $$ \mu = \mathbb{E}X = \mathbb{E}X_{1} + \cdots + \mathbb{E}X_{\sqrt{L}+\log_{2}(m)} \geq \frac{\sqrt{L}+\log_{2}(m)}{2\sqrt{e}} \enspace. $$
To estimate the probability that $ \sqrt{L} $ transmissions were heard we will use the Chernoff's inequality.

We want to have that $ (1-\epsilon)\mu = \sqrt{L} $. Thus $ \epsilon = \frac{\mu - \sqrt{L}}{\mu} = \frac{\sqrt{L} + \log_{2}(m) - 2\sqrt{e}\sqrt{L}}{\sqrt{L} + \log_{2}(m)} $. In what follows
$ 0 < \epsilon^{2} < 6 $ i.e. it is bounded.
Hence

$$ \mathbb{P}[X < \sqrt{L}] \leq e^{-\frac{\left(\frac{\sqrt{L} + \log_{2}(m) - 2\sqrt{e}\sqrt{L}}{\sqrt{L} + \log_{2}(m)}\right)^{2}}{2}  \frac{\sqrt{L} + \log_{2}(m)}{2\sqrt{e}}}
= e^{-c\;\sqrt{L} + \log_{2}(m)} \enspace, $$
for some bounded constant $ c $. We conclude that with probability $ 1 - e^{-c\;\sqrt{L} + \log_{2}(m)} $ we confirm the correct $ i $ which describes and estimates 
the current number of operational machines.
\end{proof}

\begin{lemma}
 $\mat(i, L, m) $ will \textbf{not be} executed if there are more than $ \frac{m}{2^{i-1}} $ operational machines, with probability not less than  
 $ 1 - (\log_{2}(m))^{2}\:\max\{ e^{-\frac{1}{8}\sqrt{L}} ,e^{-c\;\sqrt{L} + \log_{2}(m)} \}. $
 \label{lem46}
\end{lemma}

\begin{proof}
Let $ A_{i} $ denote an event that at the beginning of and execution of the \mat($i, L, m $) procedure there are no more than $ \frac{m}{2^{i-1}} $ operational machines.

The basic case when $ i = 0 $ is trivial, because initially we have $ m $ operational machines, thus $ \mathbb{P}(A_{0}) = 1 $. 
Consider an arbitrary~$ i $.
We know that 

$$ \mathbb{P}(A_{i}) = \mathbb{P}(A_{i}|A_{i-1})\mathbb{P}(A_{i-1}) + \mathbb{P}(A_{i}|A^{c}_{i-1})\mathbb{P}(A^{c}_{i-1}) \geq \mathbb{P}(A_{i}|A_{i-1})\mathbb{P}(A_{i-1}) \enspace. $$
Let us estimate $ \mathbb{P}(A_{i}|A_{i-1}) $. Conditioned on that event $ A_{i-1} $ holds, we know that after executing \mat($ i-1, L, m $) we had $ \frac{m}{2^{i-2}} $ operational machines.
In what follows if we are now considering \mat($i, L, m$), then we have two options:

\begin{enumerate}
 \item \mat($i-1, L, m$) returned \textit{false},
 \item \mat($i-1, L, m$) returned \textit{true}.
\end{enumerate}
We examine what these cases mean:
\begin{enumerate}
 \item[Ad 1.] If the procedure returned \textit{false} then we know from Lemma~\ref{lem44} that with probability $ 1 - e^{-c\;\sqrt{L} + \log_{2}(m)} $ there had to be no more than $ \frac{m}{2^{i-1}} $ operational
 machines. If that number would be in  $ (\frac{m}{2^{i-1}}, \frac{m}{2^{i-2}}] $ then the probability of returning \textit{false} would be less than $ e^{-c\;\sqrt{L} + \log_{2}(m)} $.
 \item[Ad 2.] If the procedure returned \textit{true}, this means that when executing it with parameters $ (i-1, L, m) $ we had no more than $ \frac{m}{2^{i-1}} $ operational machines.
 Then the internal loop of \algz\ was broken, so according to Lemma~\ref{lem42} we conclude that the overall number of operational machines had to reduce by half with probability at least
 $ 1 - e^{-\frac{1}{8}\sqrt{L}} $.
\end{enumerate}
Consequently, we deduce that $ \mathbb{P}(A_{i}|A_{i-1}) \geq (1 - \max\{ e^{-\frac{1}{8}\sqrt{L}} ,e^{-c\;\sqrt{L} + \log_{2}(m)} \}) $. Hence
$ \mathbb{P}(A_{i}) \geq (1 - \max\{ e^{-\frac{1}{8}\sqrt{L}} ,e^{-c\;\sqrt{L} + \log_{2}(m)} \})^{i} $. Together with the fact, that $ i \leq \log_{2}(m) $ and the Bernoulli inequality we have that

$$ \mathbb{P}(A_{i}) \geq 1 - \log_{2}(m)\:\max\{ e^{-\frac{1}{8}\sqrt{L}} ,e^{-c\;\sqrt{L} + \log_{2}(m)} \} \enspace.$$
We conclude that the probability that the conjunction of events $ A_{1},\cdots,A_{\log_{2}(m)} $ will hold is at least

$$ \mathbb{P}\left(\bigcap_{i=1}^{log_{2}(m)}A_{i}\right) \geq 1 - (\log_{2}(m))^{2}\:\max\{ e^{-\frac{1}{8}\sqrt{L}} ,e^{-c\;\sqrt{L} + \log_{2}(m)} \} \enspace.$$%
\end{proof}

\begin{theorem}
 \algz\ performs $\mathcal{O}(L + m\sqrt{L} + m\alpha) $ expected work against the non-adaptive adversary. 
 \label{theorem41}
\end{theorem}

\begin{proof}
In the algorithm we are constantly controlling whether condition $ \frac{m}{2^{i}} > \sqrt{L} $ holds. If not, then we execute \algx\ which
complexity is $ \mathcal{O}(L + m \alpha) $ for such parameters.

If this condition does not hold initially then we check another one i.e. whether $ \log_{2}(m) > e^{\frac{\sqrt{L}}{32}} $ holds. For such configuration we assign all the tasks to every machine.
The work accrued during such a procedure is $ \mathcal{O}(mL) $. However when $ \log_{2}(m) > e^{\frac{\sqrt{L}}{32}} $ then together with the fact that  $ e^{x} < x $ we have that $ \log_{2}(m) > L $ 
and consequently the total complexity is $ \mathcal{O}(m\log(m)) $.

Finally, the successful machines, that performed all the jobs have to confirm this fact by executing \conwor. We demand that only one machine will transmit and if this happens the algorithm terminates.
The expected value of a geometric random variable lets us assume that this confirmation will happen in expected number of $ \mathcal{O}(\log(m)) $ rounds, generating $ \mathcal{O}(m\log(m)) $ work.

When none of the conditions mentioned above hold, we proceed to the main part of the algorithm.
The testing procedure by \mat\ for each of disjoint cases, where $ M \in (\frac{m}{2^{i}}, \frac{m}{2^{i-1}}] $ requires a certain amount of work that can be bounded by
$ \mathcal{O}(m \sqrt{L} + m \log(m)) $, as there are $ \sqrt{L} + \log_{2}(m) $ testing phases in each case and at most $ \frac{m}{2^{i}} $ machines take part in a single testing phase for a certain case.

In the algorithm we run through disjoint cases where $ M \in (\frac{m}{2^{i}}, \frac{m}{2^{i-1}}] $. 
From Lemma~\ref{lem42} we know that when some of the leaders were crashed, then a proportional number of all the machines had to be crashed with high probability.
When leaders are crashed but the number of operational machines still remains in the same interval, then the lowest number of jobs will be confirmed if only the initial segment of machines will transmit.
As a result, when half of the leaders were crashed, then the system still confirms $ \frac{L}{8} = \Omega(L) $ tasks.
This means that even if so many crashes occurred, $ \mathcal{O}(1) $ epochs still suffice to do all the jobs. Summing work over all the cases may be estimated as
$ \mathcal{O}(m\sqrt{L}) $.


By Lemma~\ref{lem46} we conclude that the expected work complexity is bounded by:
$$ \left((\log(m))^{2}\:\max\{ e^{-\frac{1}{8}\sqrt{L}} ,e^{-c\;(\sqrt{L}+\log(m))} \}\right)\mathcal{O}(mL + m\sqrt{L}\log^{2}(m)) $$
$$ + \left(1 - (\log(m))^{2}\:\max\{ e^{-\frac{1}{8}\sqrt{L}} ,e^{-c\;(\sqrt{L}+\log(m))} \}\right)\mathcal{O}(m\sqrt{L} + m \log(m)) = \mathcal{O}(m\sqrt{L}) \enspace, $$
where the first expression  comes from the fact, that if we entered the main loop of the algorithm then we know that we are in a configuration where $ \log_{2}(m) \leq e^{\frac{\sqrt{L}}{32}} $. 
Thus we have that
 
$$ \frac{mL + m\sqrt{L}\log^{2}(m)}{e^{\frac{\sqrt{L}}{8}}} \leq \frac{mL + mL\log^{2}(m)}{e^{\frac{\sqrt{L}}{16}}e^{\frac{\sqrt{L}}{16}}} \leq \frac{m + m\log^{2}(m)}{e^{\frac{\sqrt{L}}{16}}} 
\leq m + m\log(m) = \mathcal{O}(m\log(m)) \enspace.  $$
\jm{The second expression follows from the fact that the algorithm was in a condition where $ \sqrt{L} \geq m $, and because it also holds 
$ m \geq \log(m) $ then consequently $ \sqrt{L} \geq \log(m) $.}

Altogether, we have $ \mathcal{O}(L + m\alpha) $ work, resulting from reaching condition $ \frac{m}{2^{i}} > \sqrt{L} $ and $ \mathcal{O}(m\log(m)) $ work resulting from reaching 
$ \log_{2}(m) \leq e^{\frac{\sqrt{L}}{32}} $. Additionally, there is $ \mathcal{O}(L + m\sqrt{L}) $ work, because of the \mat\ procedure and the corresponding effort of the leaders, so overall
the work complexity of \algz\ is $ \mathcal{O}(L + m\sqrt{L} +m\alpha) $,
what ends the proof.
\end{proof}
\section{Comparison}
\label{sec51}

In this section we carry out a brief comparison of the formulas that we proved in previous sections, i.e., the upper bound for preemptive scheduling from
Theorem~\ref{thm21} vs 
the lower bound in the model of jobs without preemption from Corollary~\ref{cor33}.
This comparison shows the range of dependencies between model parameters
for which both models are separated, i.e., the upper bound in the model with 
preemption is asymptotically smaller than the lower bound in the model without preemption.
We settle the scope on the, intuitively greater, bound for the model without preemption and examine how the ranges of the parameters influence the magnitude of the formulas. Let us
recall both formulas: 

\begin{itemize}
 \item[] Preemptive, upper bound: $ L + \mathcal{O}(m\sqrt{L} + m\min\{f, L\} + m\alpha) $;
 \item[] Non-preemptive, lower bound: $ L + \Omega(\frac{L}{n}m\sqrt{n} + \alpha m \min\{ f, n\}) $.
\end{itemize}

If $ L $ is the factor that dominates the bound in the non-preemptive model, then by simple comparison to other factors we have that $ L $ also dominates the bound in the preemptive model.
In what follows both formulas are asymptotically equal when the total number of tasks is appropriately large.

When $ \frac{L}{n}m\sqrt{n} $ dominates the non-preemptive formula, then by a simple observation we have that $ 1 \leq \sqrt{\frac{L}{n}} $, \jm{because the number of all tasks is not less than the number of jobs},
and applying a square root does not affect the inequality. Thus, multiplying both sides of the inequality by $ m\sqrt{L} $ gives us that $ m\sqrt{L} \leq \frac{L}{n}m\sqrt{n} $.
It is also easy to see that $ \alpha m\min\{ f, n\} $ from the non-preemptive formula is asymptotically greater than $ m\min\{f, L\} + m\alpha $ from the preemptive one. 

Consequently, if $ \frac{L}{n}m\sqrt{n} $ or $ \alpha m\min\{ f, n\} $ dominates the bound in the non-preemptive formula then the non-preemptive formula is asymptotically greater than the preemptive one.

On the other hand when $ L $ dominates the bound then they are asymptotically equal as both formulas linearize for such magnitude. Such results are somehow compatible with our suspicions about the two
models as intuitively the non-preemptive model for channel without collision detection is more demanding in most cases.

Note that the upper bound obtained for randomized algorithm \algz\ 
against non-adaptive adversary matches the absolute lower bound 
$L+\Omega(m\sqrt{L} + m\alpha) $
in this model, given in Lemma~\ref{lem22},
and therefore beats both the compared formulas for deterministic solutions.

\section{Conclusions}
\label{sec61}

In this paper we addressed the problem of performing jobs of arbitrary lengths on a multiple-access channel \pw{without collision detection}. In particular, we investigated two scenarios\pw{: with and without preemption}.

The first one was the model with preemption, where jobs can be done partially, even by different machines. 
Jobs with preemption are likely to be seen as chains
of tasks, that need to be done in a certain order. The natural bottleneck for such problem is therefore the \pw{length of the} longest job $ \alpha $.
We showed a lower bound for the considered problem that is $\Omega(L + m\sqrt{L} + m\min\{f, L\} + m\alpha)$ and designed an algorithm that meets the proved bound, \pw{hence settled the problem.} 
Additionally we considered a distinction between oblivious and non-oblivious jobs and showed that \pw{such a distinction does not matter with respect to the adaptive adversary}.

On the other hand we answered the question of how to deal with jobs without preemption, that need to be done in one piece in order to confirm progress. Here, as well, we showed a lower bound for the
problem and developed a solution, basing on the one for preemptive jobs, \pw{yet this} turned out \jm{not to match the lower bound that we showed}. 
The formula for work is $ L + \mathcal{O}(\frac{L}{n}m\sqrt{n} + \alpha m \min\{ f, n\})$, while
we proved the lower bound for the problem to be $ L + \Omega(\frac{L}{n}m\sqrt{n} + \frac{L}{n} m \min\{ f, n\} + m \alpha)$.

What is more, even though we may hide factor $ L $ in the asymptotic notation, it still appears with constant $ 1 $ in both solutions and represent the amount of work, resulting from the number of tasks 
needed to perform for the considered problem.

We also showed that randomization helps in case of the non-adaptive adversary,
lowering the cost of preemptive scheduling to the absolute lower bound 
$L+\Omega(m\sqrt{L} + m\alpha)$,
while it does not help against more severe adaptive adversary.

Considering the open directions for the research considered in this paper, it is  natural to address the question of channels with collision detection. It may happen that the distinction
between oblivious and non-oblivious jobs will matter when such feature is available. Furthermore it is worth considering whether randomization could help improving the \pw{results} while considering non-preemptive setting or against adversaries of intermediate power, as it took place in similar 
papers considering the Do-All problem on a shared channel \cite{KKM2017}. 

Finally, the primary goal of this work was to translate scheduling from classic
models to the model of a shared channel, in which it was not considered in depth; therefore, a natural open direction is to extend the model by other features
considered in the scheduling literature.

\bibliographystyle{plain}


\begin{thebibliography}{10}

\bibitem{Bar-NoyNS03}
Amotz Bar{-}Noy, Joseph Naor, and Baruch Schieber.
\newblock Pushing dependent data in clients-providers-servers systems.
\newblock {\em Wireless Networks}, 9(5):421--430, 2003.

\bibitem{Chl}
Bogdan~S. Chlebus.
\newblock {\em Randomized communication in radio networks, a chapter. In: Pardalos, P.M., Rajasekaran, S., Reif, J.H., Rolim, J.D.P. (eds.), Handbook on Randomized Computing.}
\newblock Kluwer Academic Publisher, Drodrecht, 2001, vol. 1, 401--456.

\bibitem{CDS}
Bogdan~S. Chlebus, Roberto De~Prisco, and Alex~A. Shvartsman.
\newblock Performing tasks on synchronous restartable message-passing
  processors.
\newblock {\em Distrib. Comput.}, 14(1):49--64, January 2001.

\bibitem{CGKS}
Bogdan~S. Chlebus, Leszek Gasieniec, Dariusz~R. Kowalski, and Alexander~A.
  Shvartsman.
\newblock Bounding work and communication in robust cooperative computation.
\newblock In {\em Proceedings of the 16th International Conference on
  Distributed Computing}, DISC '02, pages 295--310, London, UK, UK, 2002.
  Springer-Verlag.

\bibitem{CK}
Bogdan~S. Chlebus and Dariusz~R. Kowalski.
\newblock Randomization helps to perform independent tasks reliably.
\newblock {\em Random Structures and Algorithms}, 24(1):11--41, 2004.

\bibitem{CKL}
Bogdan~S. Chlebus, Dariusz~R. Kowalski, and Andrzej Lingas.
\newblock Performing work in broadcast networks.
\newblock {\em Distributed Computing}, 18(6):435--451, 2006.

\bibitem{DMY}
Roberto De~Prisco, Alain Mayer, and Moti Yung.
\newblock Time-optimal message-efficient work performance in the presence of
  faults.
\newblock In {\em Proceedings of the Thirteenth Annual ACM Symposium on
  Principles of Distributed Computing}, PODC '94, pages 161--172, New York, NY,
  USA, 1994. ACM.

\bibitem{DHW}
Cynthia Dwork, Joseph~Y. Halpern, and Orli Waarts.
\newblock Performing work efficiently in the presence of faults.
\newblock {\em {SIAM} J. Comput.}, 27(5):1457--1491, 1998.

\bibitem{GMY}
Z.~Galil, A.~Mayer, and Moti Yung.
\newblock Resolving message complexity of byzantine agreement and beyond.
\newblock In {\em Proceedings of the 36th Annual Symposium on Foundations of
  Computer Science}, FOCS '95, page 724, Washington, DC, USA, 1995. IEEE
  Computer Society.

\bibitem{Gal}
Robert~G. Gallager.
\newblock A perspective on multiaccess channels.
\newblock {\em IEEE Trans. Information Theory}, 31:124--142, 1985.

\bibitem{GKS}
Chryssis Georgiou, Dariusz~R. Kowalski, and Alexander~A. Shvartsman.
\newblock Efficient gossip and robust distributed computation.
\newblock {\em Theor. Comput. Sci.}, 347(1-2):130--166, November 2005.

\bibitem{GSbook}
Chryssis Georgiou and Alexander~A. Shvartsman.
\newblock {\em {Cooperative Task-Oriented Computing: Algorithms and
  Complexity}}.
\newblock Morgan {\&} Claypool Publishers, 2011.

\bibitem{GJW96}
Veena Gondhalekar, Ravi Jain, and John Werth.
\newblock Scheduling on {Airdisks}: Efficient access to personalized
  information services via periodic wireless data broadcast.
\newblock Technical Report TR-96-25, Department of Computer
  Science, University of Texas, Austin, TX, 1996.


\bibitem{CG91}
Coffman,Jr., E. G. and Garey, M. R.
\newblock Proof of the 4/3 Conjecture for Preemptive vs. Nonpreemptive Two-processor Scheduling,
\newblock In {\em Proceedings of the Twenty-third Annual ACM Symposium on Theory of Computing},
 STOC 1991, pages 241--248, New Orleans, Louisiana USA, 1991. ACM.


\bibitem{KKM2017}
Marek Klonowski, Dariusz~R. Kowalski, and Jaros\l{}aw Mirek.
\newblock Ordered and delayed adversaries and how to work against them on
  shared channel.
\newblock {\em CoRR}, abs/1706.08366, 2017.

\bibitem{Hoef}
Wassily Hoeffding.
\newblock Probability inequalities for sums of bounded random variables.
\newblock In {\em Journal American Statistics Association}, vol. 58,
pages 13-30, 1963.

\bibitem{KS03}
Dariusz~R. Kowalski and Alex~A. Shvartsman.
\newblock Performing work with asynchronous processors: Message-delay-sensitive
  bounds.
\newblock In {\em Proceedings of the Twenty-second Annual Symposium on
  Principles of Distributed Computing}, PODC '03, pages 265--274, New York, NY,
  USA, 2003. ACM.

\bibitem{KWZ2017}
Dariusz~R. Kowalski, Prudence W.H. Wong, and Elli Zavou.
\newblock Fault tolerant scheduling of tasks of two sizes under resource
  augmentation.
\newblock {\em J. of Scheduling}, 20(6):695-711, December 2017.

\bibitem{LST16}
Giorgio Lucarelli, Abhinav Srivastav, and Denis Trystram.
\newblock From preemptive to non-preemptive scheduling using rejections.
\newblock In {\em Preceedings of the 22nd International Conference on Computing and Combinatorics},
  COCOON 2016, pages 510--519, Ho Chi Minh City, Vietnam, 2016.
  
  
\bibitem{McN59}
Robert McNaughton.
\newblock Scheduling with deadlines and loss functions.
\newblock {\em Manage. Sci.}, 6(1):1--12, October 1959.

\end{thebibliography}

\end{document}